\documentclass[11pt]{amsart}
\usepackage{amsmath,amssymb,latexsym,dsfont}
\usepackage[small]{caption}
\usepackage{graphicx,color,mathrsfs,tikz}
\usepackage{subfigure,color}
\usepackage{cite}
\usepackage[colorlinks=true,urlcolor=blue,
citecolor=red,linkcolor=blue,linktocpage,pdfpagelabels,
bookmarksnumbered,bookmarksopen]{hyperref}
\usepackage[italian,english]{babel}
\usepackage{units}
\usepackage{enumitem}
\usepackage[left=2.65cm,right=2.65cm,top=3.2cm,bottom=3.2cm]{geometry}
\usepackage[hyperpageref]{backref}

\usepackage[colorinlistoftodos,prependcaption]{todonotes}

 \DeclareMathOperator{\imag}{{Im}}
\DeclareMathOperator{\real}{{Re}}
\renewcommand{\ll}{\left\langle}
\newcommand{\rr}{\right\rangle}

\numberwithin{equation}{section}
\newtheorem{theorem}{Theorem}
\newtheorem{proposition}{Proposition}
\newtheorem{lemma}{Lemma}
\newtheorem{remark}{Remark}

\newtheorem{definition}{Definition}
\theoremstyle{definition}

\numberwithin{equation}{section}
\numberwithin{theorem}{section}
\numberwithin{lemma}{section}
\numberwithin{proposition}{section}
\numberwithin{corollary}{section}
\numberwithin{corollary}{section}
\numberwithin{example}{section}
\numberwithin{definition}{section}
\numberwithin{remark}{section}

\renewcommand{\epsilon}{\eps}

\newcommand{\C}{\mathbb{C}}

\newcommand{\R}{{\mathbb R}}

\newcommand{\eps}{\varepsilon}

\newcommand{\pnorm}[2][]{\if #1'' \left|#2\right|_p \else \left|#2\right|_{#1} \fi}

\newcommand{\loc}{{\rm loc}}

\newcommand{\dive}{\mbox{div} \, }

\newcommand{\hu}{\hat u}
\newcommand{\hA}{\hat A}
\newcommand{\hSigma}{\hat \Sigma}
\newcommand{\D}{\mbox{Data}}
\newcommand{\bH}{{\bf H}}
\newcommand{\bA}{{\bf A}}
\newcommand{\bSigma}{{\bf \Sigma}}

\renewcommand{\theta}{\vartheta}

\def\XXint#1#2#3{{\setbox0=\hbox{$#1{#2#3}{\int}$ }
		\vcenter{\hbox{$#2#3$ }}\kern-.6\wd0}}

\newcommand{\VV}{\mathcal{V}}
\newcommand{\HH}{\mathcal{H}}
\newcommand{\LL}{{L}^2(\R^3)}
\newcommand{\Hcurl}{{H}_{{curl}}(\mathbb{R}^3)}

\newcommand{\cA}{{\hat A}}
\newcommand{\cU}{{\hat u}}

\DeclareMathOperator{\supp}{supp}
\title{Negative index materials: some mathematical perspectives}

\author[H.-M. Nguyen]{Hoai-Minh Nguyen}

\address{Hoai-Minh Nguyen \newline\indent
Department of Mathematics \newline\indent
	EPFL SB CAMA \newline\indent
	Station 8 CH-1015 Lausanne, Switzerland}
\email{hoai-minh.nguyen@epfl.ch}

\begin{document}

\begin{abstract}

Negative index materials are artificial structures whose  refractive index has a negative value over some frequency range. These materials were postulated and  investigated theoretically by Veselago  in 1964 and were confirmed experimentally by Shelby, Smith, and Schultz in 2001.  New fabrication techniques now allow for the construction of negative index materials at
scales that are interesting for applications, which has made them a very active topic of investigation. 
In this paper, we report various mathematical results on the properties of negative index materials and their applications. The topics discussed herein include  superlensing using complementary media, cloaking using complementary media,  cloaking an object via anomalous localized resonance, and the well-posedness and the finite speed propagation in media consisting of dispersive metamaterials. Some of the results have been refined and have simpler proofs than the original ones. 
\end{abstract}

\maketitle

\tableofcontents

\section{Introduction}

Negative index materials (NIMs) are artificial structures whose  refractive index has a negative value over some frequency range. These materials were postulated and  investigated theoretically by Veselago  in 1964  \cite{Veselago} and  their existence  was confirmed experimentally by Shelby, Smith, and Schultz in 2001 \cite{ShelbySmithSchultz}.  New fabrication techniques now allow for the construction of NIMs at
scales that are interesting for applications, which has made them a very active topic of investigation.  NIMs have  attracted a lot of attention from the scientific community, not only because of potentially interesting applications, but also because of challenges involved in understanding
their peculiar properties due to the sign changing coefficients in the equations modelling the phenomena.

Concerning the electromagnetic wave, wave phenomena in the time harmonic regime are modeled  by  Maxwell equations 
\begin{equation}\label{eq-electromagnetic}
\left\{ \begin{array}{cl}
\nabla \times E = i k \mu H  & \mbox{ in } \R^3, \\[6pt]
\nabla \times H = - i k \eps E + j & \mbox{ in } \R^3. 
\end{array} \right. 
\end{equation}
Here $\eps$ and  $\mu$ are $3 \times 3$ matrix-valued functions 
corresponding  to the permittivity and permeability of the medium, respectively,  $j$ is the density of charge,  $k > 0$ is the frequency, and $i$ is the pure imaginary complex number ($i^2 = -1$). NIMs fall into  the region in which both $\eps$ and $\mu$ are negative (in the matrix sense); for a  standard material, both $\eps$ and $\mu$ are positive. Concerning the acoustic wave, phenomena in the time harmonic regime are modeled  by the Helmholtz equation
\begin{equation}\label{eq-acoustic}
\dive (\bA \nabla u) + k^2 \bSigma u = f \mbox{ in } \R^d, 
\end{equation}
with $d = 2, \, 3$ where $\bA$ is a $d \times d$ matrix-valued function  and $\bSigma$ is a function that describes the properties of the medium. For NIMs, $\bA$ and $\bSigma$ are negative;  for a  standard material, both $\bA$ and $\bSigma$ are positive. In the acoustic quasistatic regime $k = 0$, the medium is therefore characterized by the matrix $\bA$. This  regime will be discussed in detail in this paper to illustrate the phenomena and  mathematical ideas used to investigate NIMs with an exception in Section~\ref{sect-Maxwell} where only the electromagnetic setting is considered. Facts  related to  the finite frequency regime and the electromagnetic setting are also mentioned. 

To correctly investigate these equations, one adds some loss (friction or damping effects) into the region of NIMs  and then studies these equations as the loss goes to 0. 
Mathematically, the study of media consisting of NIMs  faces two difficulties. First, the equations describing the wave phenomena have sign changing coefficients, hence the ellipticity and the compactness  are lost in general. Second, a localized resonance, i.e., the field (the solution) explodes in some regions and remains bounded in some others as the loss goes to 0, might appear. In this paper, we report various mathematical results on the properties of NIMs and their applications. The topics are superlensing using complementary media (Section~\ref{sect-SCM}), cloaking using complementary media (Section~\ref{sect-CCM}), cloaking {\it an object} via anomalous localized resonance (Section~\ref{sect-CALR-object}), and the well-posedness and the finite speed propagation properties of electromagnetic waves in the time domain for media consisting of dispersive metamaterials (Section~\ref{sect-Maxwell}). 
Concerning the first three topics, refined results in comparison with the original works will be presented.
Other aspects on NIMs, such as the stability of NIMs and cloaking a source via anamlous localized resonance, will be mentioned  briefly in the last section (Section~\ref{sect-perspectives})  in which we also discuss future directions. This report can be considered  as a  companion to the one in \cite{Ng-survey-NIM-1} written in 2015 in which superlensing using complementary media, cloaking using complementary media,  and the stability of NIMs and cloaking a source via anomalous localized resonance are discussed in the spirit of the original works. 

\section{Superlensing using complementary media}\label{sect-SCM}

Superlensing using NIMs  was suggested by Veselago in his seminal paper \cite{Veselago}. In this paper,   he considered  a slab lens with $\eps = \mu = - I$, where $I$ denotes the identity matrix,  and investigated its lensing property  using  ray theory. Later, the study of  cylindrical lenses in the two dimensional  quasistatic regime, the Veselago slab lens, cylindrical lenses and   spherical lenses in the finite frequency regime were respectively suggested and examined  by Nicorovici, McPhedran, and Milton  \cite{NicoroviciMcPhedranMilton94},  Pendry  \cite{PendryNegative, PendryCylindricalLenses}, and  Pendry and Ramakrishna  \cite{PendryRamakrishna0, PendryRamakrishna} for constant isotropic objects. 

In this section, we present superlensing schemes in the spirit of  \cite{Ng-Superlensing} in which we  established  superlensing using complementary media for arbitrary objects. The superlensing  schemes in \cite{Ng-Superlensing} are inspired by  the ones  suggested in \cite{NicoroviciMcPhedranMilton94, PendryNegative, PendryCylindricalLenses, PendryRamakrishna} but different from there. The lenses  in \cite{Ng-Superlensing}  also have their roots from  \cite{Ng-Complementary} in which  complementary media were defined and investigated  from mathematical point of views. It was shown later in \cite{Ng-CALR-object} that the modification proposed in \cite{Ng-Superlensing} in comparison with \cite{NicoroviciMcPhedranMilton94, PendryNegative, PendryCylindricalLenses, PendryRamakrishna} was necessary in order to achieve superlensing (see also Section~\ref{sect-CALR-object}).

We next mathematically  describe superlensing using complementary media. Let  $B_r$ denote the ball in $\R^d$ centered at the origin and of radius $r > 0$.  We first consider the quasistatic acoustic setting in  a  two dimensional, bounded domain.  
To magnify $m$ times ($m>1$)  an {\it arbitrary} object in $B_{r_0}$ with  $r_0 > 0$, one uses a  lens  consisting  of two layers as follows. The first layer in $B_{r_1} \setminus B_{r_0}$ is characterized by the identity matrix $I$,  and the second one in $B_{r_2} \setminus B_{r_1}$ is characterized by the matrix $-I$. Here $r_1$ and $r_2$ are defined by 
\begin{equation}\label{choice-r12}
r_1 = m^{1/2} r_0 \quad \mbox{ and } \quad  r_2 =  m r_0.
\end{equation}
Different choices for $r_1$ and $r_2$ are possible. Nevertheless, there are some restrictions on them. In particular, $r_1/r_0$ cannot be too close to 1 (see Remark~\ref{rem-r1-r2}). The choice considered in \eqref{choice-r12} has the advantage that the system is somehow  stable for small loss (see \eqref{SCM-key-point-1}) and our proof of  superlensing is quite simple in this case.

Assume that the object inside $B_{r_0}$ is characterized by a  symmetric, uniformly elliptic, matrix-valued function $a$, i.e., for some constant $\Lambda \ge 1$, 
$$
\Lambda^{-1} |\xi|^2 \le a(x) \xi \cdot \xi  \le \Lambda |\xi|^2 \mbox{ for a.e. } x \in B_{r_0}  \mbox{ and for } \xi \in \R^2, 
$$
and the medium outside $B_{r_2}$ is homegeneous and, hence, is  characterized by the identity matrix $I$. Then,  
with the loss being described by a small, positive parameter $\delta$, the medium considered is characterized by $A_\delta := s_\delta A$, \footnote{$A_0$ plays the role of $\bA$ in \eqref{eq-acoustic}.}
where 
\begin{equation}\label{SCM-defA-ss}
A = \left\{ \begin{array}{cl} a & \mbox{ in } B_{r_0},\\[6pt]
I  & \mbox{ otherwise}, 
\end{array} \right. 
\quad \mbox{ and } \quad  
s_\delta = \left\{ \begin{array}{cl} -1 - i \delta  & \mbox{ in } B_{r_2} \setminus B_{r_1}, \\[6pt]
1 & \mbox{ otherwise},
\end{array} \right.   \mbox{ for } \delta \ge 0. 
\end{equation}

The superlensing property of the considered lens  is confirmed by the following theorem. 
\begin{theorem}\label{SCM-thm1} Let  $0 < \delta < 1$, $\Omega$ be a smooth, bounded, open subset of $\R^2$,   $f \in L^{2}(\Omega)$,  and set $r_3 = r_2^2/ r_1$. Assume that $B_{r_3} \subset \subset \Omega$ and   $\supp f \subset \Omega \setminus B_{r_{3}}$, and let $u_\delta \in H^{1}_0(\Omega)$ be the unique solution of the equation 
$$
\dive(A_\delta \nabla u_\delta ) = f \mbox{ in } \Omega.
$$
We have 
\begin{equation}\label{SCM-key-point-1}
\| u_{\delta} -  \hu \|_{H^1(\Omega)} \le C \delta^{1/2} \| f\|_{L^2(\Omega)}
\end{equation}
and 
\begin{equation}\label{SCM-key-point-2}
\| u_{\delta} -  \hu \|_{H^1(\Omega \setminus B_{r_3})} \le C \delta \| f\|_{L^2(\Omega)},  
\end{equation}
for some positive constant $C$ independent of $f$ and $\delta$.  In particular,  
\begin{equation}\label{SCM-key-point-3}
u_\delta \to \hu \mbox{ in } H^1(\Omega \setminus B_{r_3}) \mbox{ as } \delta \to 0. 
\end{equation}
Here $\hu \in H^1_0(\Omega)$ is the unique solution of the equation 
\begin{equation*}
\dive (\hA \nabla \hu ) = f \mbox{ in } \Omega, 
\mbox{ where } 
\hA = \left\{ \begin{array}{cl} a( \cdot /m) & \mbox{ in } B_{r_2},\\[6pt]
I  & \mbox{ otherwise}. 
\end{array} \right. 
\end{equation*}
\end{theorem}

%
%

\begin{proof} We first prove \eqref{SCM-key-point-1}. 
The key idea of its proof is to  construct a solution $u_0 \in H^1_0(\Omega)$ to  the equation $\dive (A_0 \nabla u_0) = f$ in $\Omega$. To motivate the construction of $u_0 \in H^1_0(\Omega)$ below, we first assume that there exists such a solution $u_0$. 
Let $u_{1, 0}$ be the reflection of $u_0$ in $B_{r_2}$ through $\partial B_{r_2}$ via the Kelvin transform $F$ associated with $\partial B_{r_2}$, i.e., 
\begin{equation}
u_{1, 0} (x) = u_0 \circ F^{-1} \mbox{ for } x \in \R^2 \setminus B_{r_2}, \mbox{ where } F(x) := r_2^2 x / |x|^2 \mbox{ for } x \in \R^2. 
\end{equation}
Note that $F$ respectively transforms $\partial B_{r_1}$ onto $\partial B_{r_3}$ and $\partial B_{r_0}$ onto $\partial B_{r_2}$; the constant $r_3$ appears naturally here.  
Since $\Delta u_{0} = 0$ in $B_{r_2} \setminus \bar B_{r_1}$ and in $B_{r_1}$, it follows that 
$$
\Delta u_{1, 0} =  0 \mbox{ in } B_{r_3} \setminus \bar B_{r_2} \mbox{ and in } \R^2 \setminus \bar B_{r_3}. 
$$ 
Moreover, by taking into account the continuity $u_0$ and its flux on $\partial B_{r_2}$ and $\partial B_{r_3}$,  we have 
$$
u_{1, 0} = u_{0} \quad \mbox{ and } \quad  \partial_r u_{1, 0} = -  \partial_r u_{0} |_{int} = \partial_r u_{0} |_{ext} \mbox{ on } \partial B_{r_2}
$$
and 
$$
u_{1, 0} |_{ext} = u_{1, 0} |_{int}   \quad \mbox{ and } \quad \partial_r u_{1, 0} |_{ext}  = - \partial_r u_{1, 0} |_{int}  \mbox{ on } \partial B_{r_3}. 
$$
Here and in what follows, for a smooth, bounded, open subset $D$ of $\R^d$, on its boundary $\partial D$,  $u |_{ext}$ and $u |_{int} $ denote the restriction of $u$ in $\R^d \setminus \bar D$ and the restriction of $u$ in $D$, respectively,  for an appropriate function $u$.  We also denote $[u]$ as the quantity $u |_{ext} - u |_{int}$ on $\partial D$ and use similar notations for $A \nabla u \cdot \nu$ for an appropriate function $u$ where $A$ is a matrix and $\nu$ is the unit normal vector on $\partial D$ directed to the exterior of $D$. 

Let $u_{2, 0}$ be the reflection of $u_{1, 0}$ in $B_{r_3}$ through $\partial B_{r_3}$ via the Kelvin transform $G$ associated with $\partial B_{r_3}$, i.e., 
\begin{equation}
u_{2, 0} (x) = u_{1, 0} \circ G^{-1} \mbox{ for } x \in B_{r_3}, \mbox{ where } G(x) := r_3^2 x / |x|^2 \mbox{ for } x \in \R^2. 
\end{equation}
We then have 
\begin{equation}\label{eq-u2}
\dive (\hA \nabla  u_{2, 0} ) = 0 \mbox{ in } B_{r_3}
\end{equation}
$$
u_{2, 0}  = u_{1, 0} |_{int}   \quad \mbox{ and } \quad \partial_r u_{2, 0}  =  \partial_r u_{1, 0} |_{int}  \mbox{ on } \partial B_{r_3}. 
$$
The definition of $\hA$ in $B_{r_3}$ appears naturally here. 
Since $\hat A = I$ in $B_{r_3} \setminus B_{r_2}$ by the choices of $r_1$ and $r_2$ ($G\circ F(\partial B_{r_0}) = \partial B_{m r_0}  = \partial B_{r_2}$), it follows from \eqref{eq-u2} that 
$$
\Delta u_{2, 0} = 0 \mbox{ in } B_{r_3}. 
$$
Set 
\begin{equation}
w_{0} = \left\{ \begin{array}{cl} u_{0} &  \mbox{ in } \Omega \setminus \bar B_{r_2},\\[6pt]
u_{2, 0} & \mbox{ in } B_{r_2}. 
\end{array} \right.
\end{equation}
Then 
\begin{equation}
\dive (\hA \nabla w_0) = f \mbox{ in } \Omega \setminus \partial B_{r_2}, \quad [w_0] = 0 \mbox{ on } \partial B_{r_2}, \quad \mbox{ and } \quad [\hA \nabla w_0 \cdot \nu] = 0 \mbox{ on } \partial B_{r_2}.
\end{equation}
It follows that $w_0 \in H^1_0(\Omega)$ is a solution of 
\begin{equation}
\dive (\hA \nabla w_0) = f \mbox{ in } \Omega. 
\end{equation}
We derive that 
$$
\hu = w_0 \mbox{ in } \Omega. 
$$

Inspired by the heuristic arguments above, we define 
\begin{equation}\label{def-u0}
u_0 = \left\{\begin{array}{cl} \hu &  \mbox{ in } \Omega \setminus B_{r_2}, \\[6pt]
\hu \circ F & \mbox{ in } B_{r_2} \setminus B_{r_1}, \\[6pt]
\hu \circ G \circ F = \hu(m \, \cdot ) & \mbox{ in } B_{r_0}.
\end{array} \right. 
\end{equation}
It is clear from the definition of $\hu$ that 
\begin{equation}\label{u0-p1}
\Delta u_0 = f \mbox{ in } (\Omega \setminus \bar B_{r_2}) \cup (B_{r_2} \setminus \bar B_{r_1}) \quad \mbox{ and } \quad \dive(a \nabla u_0) = 0 \mbox{ in } B_{r_1}. 
\end{equation}
Moreover, one can check that 
\begin{equation}\label{u0-p2}
[u_0] = 0 \quad \mbox{ and } \quad [s_0 A \nabla u_0 \cdot \nu] = 0 \mbox{ on } (\partial B_{r_2} \cup \partial B_{r_1}).
\end{equation}
Combining \eqref{def-u0}, \eqref{u0-p1}, and \eqref{u0-p2}  yields that $u_0 \in H^1_0(\Omega)$ is a solution of the equation $\dive (A_0 \nabla u_0) = f$ in $\R^2$.

We have 
\begin{equation}\label{eq-diff}
\dive \big(A_\delta \nabla (u_\delta - u_0) \big) = - \dive \big((A_\delta - A_0) \nabla u_0 \big) = i \delta \dive (\chi_{B_{r_2} \setminus B_{r_1}}\nabla u_0) \mbox{ in } \Omega, 
\end{equation}
where $\chi_{D}$ denotes the characteristic function of a given subset $D$ of $\R^2$. Applying \eqref{stability-1} of Lemma~\ref{lem-stability} below, we have 
$$
\|u_\delta - u_0 \|_{H^1(\Omega)} \le C  \| \nabla u_0\|_{L^2(B_{r_2} \setminus B_{r_1})} \le C \| f\|_{L^2(\Omega)}.  
$$
This yields, by \eqref{eq-diff} and \eqref{stability-1} of Lemma~\ref{lem-stability} again, 
$$
\|u_\delta - u_0 \|_{H^1(\Omega)} \le C  \delta^{1/2}  \| f\|_{L^2(\Omega)}, 
$$
which is \eqref{SCM-key-point-1}. 

\medskip 
We next establish \eqref{SCM-key-point-2}.  Similar to the definition of $u_{1, 0}$ and $u_{2, 0}$, we define $u_{1, \delta}$ in $\R^2 \setminus \bar B_{r_2}$ and $u_{2, \delta}$ in $B_{r_3}$ as follows  
\begin{equation*}
u_{1, \delta} = u_{\delta} \circ F^{-1} \mbox{ in } \R^2 \setminus \bar B_{r_2} \quad \mbox{ and } \quad u_{2, \delta} = u_{1, \delta} \circ G^{-1} \mbox{ in } B_{r_3}. 
\end{equation*}
As above, one can verify that 
\begin{equation}\label{SCM-t1}
\Delta u_{1, \delta} = 0 \mbox{ in } B_{r_3} \setminus \bar B_{r_2}, \quad u_{1, \delta} = u_{\delta} \mbox{ on } \partial B_{r_2}, \quad (1 + i \delta) \partial_r u_{1, \delta} = \partial_r u_{\delta} |_{ext} \mbox{ on } \partial B_{r_2}, 
\end{equation}
\begin{equation}\label{SCM-t2}
\dive (\hA \nabla u_{2, \delta}) = 0 \mbox{ in } B_{r_3}, \quad u_{2, \delta} = u_{1, \delta} \mbox{ on } \partial B_{r_3}, \quad \mbox{ and } \quad  \partial_r u_{2, \delta} = (1 + i \delta) \partial_r u_{1, \delta} |_{int} \mbox{ on } \partial B_{r_3}. 
\end{equation}

Define
\begin{equation}\label{def-hu-lens}
\hu_\delta = \left\{\begin{array}{cl} u_\delta & \mbox{ in } \Omega \setminus B_{r_3}, \\[6pt]
u_\delta - (u_{1, \delta} - u_{2, \delta}) & \mbox{ in } B_{r_3} \setminus B_{r_2}, \\[6pt]
u_{2, \delta} & \mbox{ in } B_{r_2}. 
\end{array}\right. 
\end{equation}
One can check that 
\begin{equation*}
\dive\big(\hA \nabla (\hu_\delta - \hu_0)  \big) = 0 \mbox{ in } \Omega \setminus (\partial B_{r_2} \cup \partial B_{r_3}). 
\end{equation*}
Moreover, by \eqref{SCM-t1} and \eqref{SCM-t2}, one has 
\begin{equation*}
[\hu_\delta - \hu_0] = u_{\delta} - u_{1, \delta} = 0,  \quad [\hat A \nabla (u_\delta - \hu_0) \cdot e_r] = \hA \nabla (u_{\delta}|_{ext} - u_{1, \delta}) \cdot e_r  = i \delta \partial_r u_{1, \delta} |_{int}  \quad \mbox{ on } \partial B_{r_2},
\end{equation*}
and
\begin{equation*}
[\hu_\delta - \hu_0] = u_{1, \delta} - u_{2, \delta} = 0,  \quad [\hat A \nabla (u_\delta -  \hu_0) \cdot e_r] = \partial_r (u_{1, \delta}|_{int} - u_{2, \delta})   = - i \delta \partial_r u_{1, \delta} |_{int}  \quad \mbox{ on } \partial B_{r_3}. 
\end{equation*}
From Lemma~\ref{lem-stability-S} below, it follows that
\begin{equation*}
\| \hu_\delta - \hu_0 \|_{H^1(\Omega \setminus (\partial B_{r_2} \cup \partial B_{r_3}))} \le C \Big( \|\delta \partial_r u_{1, \delta} |_{int} \|_{H^{-1/2}(\partial B_{r_3})} +  \|\delta \partial_r u_{1, \delta} |_{ext} \|_{H^{-1/2}(\partial B_{r_2})} \Big) \le C \delta \| f\|_{L^2(\Omega)}. 
\end{equation*}
In the last inequality, we use \eqref{SCM-key-point-1}. Since $\hu_\delta = u_\delta$ in $\Omega \setminus \bar B_{r_{3}}$, assertion~\eqref{SCM-key-point-2} follows. 

\medskip 
The proof is complete. 
\end{proof}

\begin{remark} \rm Assertion \eqref{SCM-key-point-3}  in  a more general setting,  the setting of complementary media,  is given in \cite{Ng-Complementary}. In \cite{Ng-Complementary}, $s_\delta$ is defined by  $-1 + i \delta$ in $B_{r_2} \setminus B_{r_1}$; nevertheless, this point is not essential. The proof of \eqref{SCM-key-point-1} also has its roots from \cite{Ng-Complementary}. 
The idea is to use reflections to derive Cauchy's problems from the original equation with sign changing coefficients and then use the unique continuation principle, see, e.g., \cite{Protter60}. This can be applied  for a general structure via the change of variables rule, see Lemma~\ref{lem-TO} below. Assertion   \eqref{SCM-key-point-2} is new in comparison with \cite{Ng-Complementary} whose method only yields  $\delta^{1/2}$ instead of $\delta$ as the rate of the convergence. The key ingredient in the proof is the introduction of the auxiliary function $\hu_\delta$. This auxiliary function  was introduced in the technique of  removing localized singularity by the author to handle the localized resonance associated with NIMs in  cloaking and superlensing applications, see our previous work \cite{Ng-Negative-Cloaking, Ng-Superlensing},  Section~\ref{sect-CCM}, and Remark~\ref{rem-r1-r2}). Interestingly, it is also useful  even in stable cases for improving the convergence rate. The motivation of \eqref{SCM-key-point-2} comes from simulations obtained in the master project of Droxler at EPFL under the supervision of Hesthaven and the author. 
\end{remark}

\begin{remark} \label{rem-r1-r2} \rm The choice of $r_1$ and $r_2$ in \eqref{choice-r12} is not strict for ensuring \eqref{SCM-key-point-3}. In previous work \cite{Ng-Superlensing}, we showed that it is possible to choose 
\begin{equation*}
r_1 = m^{1/4} r_0 \quad \mbox{ and } \quad r_2 = m^{1/2} r_1. 
\end{equation*}
In fact, the approach in \cite{Ng-Superlensing} also works for the choice 
\begin{equation}\label{SCM-tt3}
r_1 \ge m^{1/4} r_0 \quad \mbox{ and } \quad  r_2 = m^{1/2} r_1. 
\end{equation} 
Instead of introducing $\hu$ as in \eqref{def-hu-lens}, we define $\hu_\delta$ as follows 
\begin{equation}\label{SCM-hu-delta}
\hu_\delta = \left\{\begin{array}{cl} u_\delta & \mbox{ in } \Omega \setminus B_{r_3}, \\[6pt]
u_\delta - (u_{1, \delta} - u_{2, \delta}) & \mbox{ in } B_{r_3} \setminus B_{m r_0}, \\[6pt]
u_{2, \delta} & \mbox{ in } B_{m r_0}. 
\end{array}\right. 
\end{equation}
Recall that, if $v \in H^1(B_{R_3} \setminus B_{R_1})$ satisfies $\Delta v = 0$ in $B_{R_3} \setminus B_{R_1}$ for $0 < R_1 < R_2 < R_3$, then  
\begin{multline*}
\| v\|_{H^{1/2}(\partial B_{R_2})} + \| \partial_r v\|_{H^{-1/2}(\partial B_{R_2})} \\[6pt] \le C \left(\| v\|_{H^{1/2}(\partial B_{R_1})} +  \| \partial_r v\|_{H^{-1/2}(\partial B_{R_1})} \right)^{\alpha} \left(\| v\|_{H^{1/2}(\partial B_{R_3})} +  \| \partial_r v\|_{H^{-1/2}(\partial B_{R_3})} \right)^{1 - \alpha},\end{multline*}
with $ \alpha = \ln (R_3/ R_2)/ \ln (R_3/ R_1)$ \footnote{This inequality can be obtained from the following representation of $v$ in $B_{R_3} \setminus B_{R_1}$: 
$$
v(r, \theta) = a_0 + b_0 \ln r + \sum_{n = 1}^\infty \sum_{\pm} ( a_{n, \pm} r^n + b_{n, \pm} r^{-n} ) e^{\pm i n \theta}  \mbox{ in } B_{R_3} \setminus B_{R_1}. $$
See also \cite[Lemma 6]{Ng-Cloaking-Maxwell}.}.  Using this inequality, one can prove that 
$$
\| u_{1, \delta} - u_\delta \|_{H^{1/2}(\partial B_{m r_0})} + \| \partial_r (u_{1, \delta} - u_\delta) \|_{H^{-1/2}(\partial B_{m r_0})} \le C \delta^{\alpha} \| v \|_{H^1(\Omega \setminus B_{r_3})}, 
$$
with $\alpha = \ln (r_3/r_2) / \ln (r_3/ r_1)$ which is greater than or equal to $1/ 2$ by \eqref{SCM-tt3} and the fact that $r_3 = r_2^2/r_1$. Applying the approach used in the proof of \eqref{SCM-key-point-1}, one can reach \eqref{SCM-key-point-3} in the case in which $\alpha > 1/2$, which is equivalent to $r_1 > m^{1/4} r_0$. The case in which $\alpha = 1/2$, corresponding to the choice $r_1 = m^{1/4} r_0$, requires further arguments; in this case, the convergence in \eqref{SCM-key-point-3} is replaced by the weak convergence.  The interested reader can find the details in \cite[the proof of (2.36)]{Ng-Superlensing}. 
\end{remark}

In the proof of Theorem~\ref{SCM-thm1}, we used the following stability result on $u_\delta$. 

\begin{lemma}\label{lem-stability} Let $d \ge 2$, $\delta_0 > 0$, $0< r_1 < r_2$,   $\Omega$ be a smooth, open subset of $\R^d$ with $B_{r_2} \subset \subset \Omega$, let $A$ be a uniformly elliptic, matrix-valued function defined in $\Omega$, and let 
$g \in H^{-1}(\Omega)$ \footnote{$H^{-1}(\Omega)$ denotes the dual space of $H^1_0(\Omega)$.}. Set $A_\delta = s_\delta A$, where $s_\delta$ is defined in  
\eqref{SCM-defA-ss}. For $0 < \delta < \delta_0$, there exists a unique solution $v_\delta \in H^1_0(\Omega)$ of
\begin{equation*}
\dive (A_\delta \nabla v_\delta) = g  \mbox{ in } \Omega. 
\end{equation*}
Moreover,
\begin{equation}\label{stability-1}
\| v_\delta \|_{H^1(\Omega)}^2 \le  \frac{C}{\delta} \left| \int_\Omega g \bar v_\delta \right|
\end{equation}
and 
\begin{equation}\label{stability-1-2}
\| v_\delta \|_{H^1(\Omega)}^2 \le  \frac{C}{\delta} \left| \Im \int_\Omega g \bar v_\delta \right| + C  \| g\|_{L^2(\Omega)}^2. 
\end{equation}
Here  $C$ denotes a positive constant  independent of $g$ and $\delta$.
\end{lemma}

Here and in what follows, for a complex number $z$, we denote $\Im z$ and $\Re z$ as the imaginary part and the real part of $z$, respectively. 

\begin{remark} \rm Various variants of Lemma~\ref{lem-stability} are used in the study of NIMs, see,  e.g., \cite{Ng-Complementary, Ng-CALR-object}. In inequality~\eqref{stability-1-2}, one only considers the imaginary part of $\int_\Omega g \bar v_\delta$. This is useful for improvements on the convergent rate of cloaking effects considered later in Sections~\ref{sect-CCM} and \ref{sect-CALR-object}. 
Nevertheless, the proof presented below is quite standard and in the same spirit. 
\end{remark}

\begin{proof} Multiplying the equation of $v_\delta$ by $\bar v_\delta$ (the conjugate of $v_\delta$), integrating by parts, and considering the imaginary part and the real part of the obtained expression, one has 
$$
\| \nabla v_\delta \|_{L^2(\Omega)}^2 \le  \frac{C}{\delta} \left| \int_\Omega g \bar v_\delta \right|.
$$
This implies \eqref{stability-1} by the Poincar\'e inequality. 

To obtain \eqref{stability-1-2}, we proceed as follows. Multiplying the equation of $v_\delta$ by $\bar v_\delta$, considering the imaginary part, one has 
\begin{equation}\label{stability-1-p1}
\| \nabla v_\delta \|_{L^2(B_{r_2} \setminus B_{r_1})}^2 \le  \frac{C}{\delta} \left| \Im \int_\Omega g \bar v_\delta \right|.
\end{equation}
We claim that 
\begin{equation}\label{claim-sta-1}
\| v_\delta\|_{L^2(B_{r_2} \setminus B_{r_1})} \le C \big( \| \nabla v_\delta \|_{L^2(B_{r_2} \setminus B_{r_1})} + \| g\|_{L^2(\Omega)} \big). 
\end{equation}
Assuming this, we obtain 
$$
\| v_\delta\|_{H^1(B_{r_2} \setminus B_{r_1})} \le C \big( \| \nabla v_\delta \|_{L^2(B_{r_2} \setminus B_{r_1})} + \| g\|_{L^2(\Omega)} \big). $$
This implies, by the trace theory,  
$$
\| v_\delta \|_{H^{1/2}(\partial B_{r_2} \cup \partial B_{r_1})} \le C \big( \| \nabla v_\delta \|_{L^2(B_{r_2} \setminus B_{r_1})} + \| g\|_{L^2(\Omega)} \big). 
$$
Using the equation of $v_\delta$ in $\Omega \setminus B_{r_3} $ and in $B_{r_1}$, we derive from the standard theory of elliptic equations that 
$$
\| v_\delta \|_{H^{1}((\Omega \setminus B_{r_2})) \cup B_{r_1})} \le C \big( \| \nabla v_\delta \|_{L^2(B_{r_2} \setminus B_{r_1})} + \| g\|_{L^2(\Omega)} \big), 
$$
and the conclusion follows from \eqref{stability-1-p1}. 

It remains to prove \eqref{claim-sta-1}, which we establish  by contradiction. Suppose that there exist a sequence $\delta_n \to 0$ (by \eqref{stability-1}) and a sequence $g_n \to 0$ in $L^2(\Omega)$  such that 
\begin{equation}\label{claim-sta-2}
1 = \| v_{\delta_n}\|_{L^2(B_{r_2} \setminus B_{r_1})} \ge n  \big( \| \nabla v_{\delta_n} \|_{L^2(B_{r_2} \setminus B_{r_1})} + \| g_n\|_{L^2(\Omega)} \big), 
\end{equation}
where $v_{\delta_n}$ is the solution corresponding to $\delta_n$ and $g_n$. 
By the trace theory, one has
$$
\| v_{\delta_n} \|_{H^{1/2}(\partial B_{r_2} \cup \partial B_{r_1})} \le C
$$
for some positive constant $C$ independent of $n$. This in turn implies that 
$$
\| v_{\delta_n} \|_{H^{1}(\Omega)} \le C. 
$$
Without loss of generality, one can assume that $v_{\delta_n}$ converges to $v_0 \in H^1_0(\Omega)$ weakly  in $H^1(\Omega)$ and strongly  in  $L^2(B_{r_2} \setminus B_{r_1})$. Moreover, 
$$
\dive(A_0 \nabla v_0) = 0 \mbox{ in } \Omega \quad \mbox{ and } \quad  v_0 \mbox{ is constant in } B_{r_2} \setminus B_{r_1}. 
$$
Since, by multiplying the equation of $v_0$ with $\bar v_0$ and integrating by parts,  
$$
\int_{\Omega} A_0 \nabla v_0 \cdot \nabla v_0 = 0,  
$$
and $v_0$ is constant in $B_{r_2} \setminus B_{r_1}$, it follows that 
$$
\int_{\Omega} |\nabla v_0|^2 = 0. 
$$
We derive that  $v_0 = 0$ in $\Omega$ since $v_0 \in H^1_0(\Omega)$. This contradicts the fact that $\int_{B_{r_2} \setminus B_{r_1}} |v_0|^2 = \lim_{n \to + \infty} \int_{B_{r_2} \setminus B_{r_1}} |v_{\delta_n}|^2 = 1$. 
\end{proof}

The following lemma is standard and was used in the proof of Theorem~\ref{SCM-thm1}. 

\begin{lemma} \label{lem-stability-S}  Let $d =2, 3$,  $\Omega$ be a smooth, open subset of $\R^d$, and let $A$ be a symmetric, uniformly elliptic,  matrix-valued function defined in $\Omega$, and let $f \in L^2(\Omega)$. 
Let $D \subset \subset  \Omega$ be a smooth, bounded, open subset of $\R^d$, let $g \in H^{1/2}(\partial D)$, and $h \in H^{-1/2}(\partial D)$. Assume that  $v \in H^1(\Omega \setminus \partial D)$ satisfies  
\begin{equation*}
\left\{\begin{array}{cl}
\dive (A \nabla v) = f & \mbox{ in } \Omega \setminus \partial D, \\[6pt]
[v] = g \mbox{ and } [A \nabla v \cdot \nu] = h & \mbox{ on } \partial D, \\[6pt]
v = 0 & \mbox{ on } \partial \Omega. 
\end{array}\right. 
\end{equation*}
Then 
\begin{equation*}
\| v\|_{H^1(\Omega \setminus \partial D)} \le C \left( \| f \|_{L^2(\Omega)} + \| g\|_{H^{1/2}(\partial D)} + \| h \|_{H^{-1/2}(\partial D)} \right), 
\end{equation*}
for some positive constant $C$ depending only on $D$, $\Omega$,  and the ellipticity of $A$. 
\end{lemma}

The approach used in the proof of Theorem~\ref{SCM-thm1} can be extended to the finite frequency regime as well as higher dimensions. The additional tool  is the following  change of variables rule, see, e.g., \cite[Lemma 2]{Ng-Complementary}.  

\begin{lemma}\label{lem-TO} Let $d \ge 2$,  $D_1 \subset \subset D_2 \subset \subset D_3$  be three smooth, bounded, open subsets of $\R^d$. Let $a \in [L^\infty(D_2 \setminus D_1)]^{d \times d}$,  $\sigma \in L^\infty(D_2 \setminus D_1)$, and let  ${\mathcal T}$ be a bijective  from $D_2 \setminus \bar D_1$ onto $D_3 \setminus \bar D_2$ such that ${\mathcal T} \in C^1(\bar{D_2} \setminus D_1)$ and ${\mathcal T}^{-1} \in C^1(\bar D_3 \setminus D_2)$. Assume that   $u \in H^1(D_2 \setminus D_1)$ and set $v = u \circ {\mathcal T}^{-1}$. Then
\begin{equation*}
\dive (a \nabla u) +  \sigma u = f \mbox{ in } D_2 \setminus D_1, 
\end{equation*}
for some $f \in L^2(D_2 \setminus D_1)$,  if and only if
\begin{equation}\label{TO-eq}
\dive ({\mathcal T}_*a \nabla v) + {\mathcal T}_*\sigma v = {\mathcal T}_* f \mbox{ in } D_3 \setminus D_2. 
\end{equation}
Assume in addition that ${\mathcal T}(x) = x$ on $\partial D_2$. Then  
\begin{equation}\label{TO-bdry}
v = u \quad \mbox{ and } \quad {\mathcal T}_*a \nabla v \cdot \nu = - a \nabla u \cdot \nu  \mbox{ on } \partial D_2.
\end{equation}
Here 
\begin{equation}\label{def-F*}
{\mathcal T}_*a(y) = \frac{D{\mathcal T}(x) a(x) \nabla {\mathcal T}(x)^T}{|\det \nabla {\mathcal T}(x)|} \quad \mbox{ and } \quad {\mathcal T}_*\sigma(y) = \frac{ \sigma(x) }{|\det \nabla {\mathcal T}(x)|}, \quad  \mbox{ where } x = {\mathcal T}^{-1}(y). 
\end{equation}
\end{lemma}

Let $a$  be a symmetric, uniformly elliptic, matrix-valued function  and $\sigma$ be a bounded complex function both defined in $B_{r_0}$ such that $\Re \sigma > c > 0$ and $\Im \Sigma \ge 0$ in $B_{r_0}$ for some $c > 0$. Assuming \eqref{choice-r12},  
we have the following result which is a variant of Theorem~\ref{SCM-thm1} in the finite frequency regime in both two and three dimensions. 

\begin{theorem}\label{SCM-thm2}
Let  $d=2, \, 3$, $0 < \delta < 1$, $k> 0$,  $R_0 > r_3$,   $f \in L^{2}(\R^d)$,  and set $r_3 = r_2^2/ r_1$. Assume that 
$\supp f \subset B_{R_0}\setminus B_{r_{3}}$,   and let $u_\delta$ be the unique outgoing solution of the equation 
$$
\dive(A_\delta \nabla u_\delta ) + k^2 \Sigma_\delta = f \mbox{ in } \R^d, 
$$ 
where $(A_\delta, \Sigma_\delta) = (s_\delta A, s_\delta \Sigma)$ and 
\begin{equation}\label{SCM-defA-ss-2}
A, \Sigma  = \left\{ \begin{array}{cl} a, \sigma & \mbox{ in } B_{r_0},\\[6pt]
F^{-1}_*I, F^{-1} _*1 &  \mbox{ in } B_{r_2} \setminus B_{r_1}, \\[6pt]
I , 1 & \mbox{ otherwise}, 
\end{array} \right. 
\quad  \mbox{ and } \quad 
s_\delta = \left\{ \begin{array}{cl} -1 - i \delta  & \mbox{ in } B_{r_2} \setminus B_{r_1}, \\[6pt]
1 & \mbox{ otherwise}.
\end{array} \right. 
\end{equation}
We have 
\begin{equation}\label{SCM-key-point-1-thm2}
\| u_{\delta} -  \hu \|_{H^1(B_R)} \le C_R \delta \| f\|_{L^2(\R^d)}
\end{equation}
and 
\begin{equation}\label{SCM-key-point-2-thm2}
\| u_{\delta} -  \hu \|_{H^1(B_R \setminus B_{r_3})} \le C_R \delta \| f\|_{L^2(\R^d)},  
\end{equation}
for some positive constant $C_R$ independent of $f$ and $\delta$. 
In particular,
$$
u_\delta \to \hu \mbox{ in } H^1_{\loc}(\R^d \setminus B_{r_3}) \mbox{ as } \delta \to 0. 
$$
Here, $\hu$ is the unique outgoing  solution of the equation 
\begin{equation*}
\dive (\hA \nabla \hu ) + k^2 \hSigma \hu= f \mbox{ in } \R^d, 
\mbox{ where } 
\hA, \hSigma= \left\{ \begin{array}{cl} m^{2-d}a(x/m), m^{-d} \sigma(x/m) & \mbox{ in } B_{r_2},\\[6pt]
I , 1 & \mbox{ otherwise}. 
\end{array} \right. 
\end{equation*}
\end{theorem}

Recall that a solution $v \in H^1_{\loc}(\R^d \setminus B_{R})$ of the equation
\begin{equation*}
\Delta v + k^2 v = 0 \mbox{ in } \R^d \setminus B_{R}, 
\end{equation*} 
for some $R>0$, is said to satisfy the outgoing condition if 
\begin{equation*}
\partial_r v - i k v = o(r^{-\frac{d-1}{2}}) \mbox{ as } r = |x| \to + \infty. 
\end{equation*}

\begin{proof} The proof of Theorem~\ref{SCM-thm2} is similar to the one of Theorem~\ref{SCM-thm1} by using Lemma~\ref{lem-TO} and applying variants of Lemmas~\ref{lem-stability} and \ref{lem-stability-S}, see,  e.g., \cite[Lemma 1]{Ng-Complementary} or \cite[Lemma 2.1]{Ng-CALR-frequency} for  variants of Lemma~\ref{lem-stability}.  The details are left to the reader. 
\end{proof}

\begin{remark} \rm Superlensing using complementary media is justified mathematically  for the electromagnetic wave  \cite{Ng-Superlensing-Maxwell}. The idea of using reflections is also useful in establishing superlensing using hyperbolic metamaterials, an interesting type of metamaterials,  \cite{BonnetierNguyen}
\end{remark}

\begin{remark} \rm Using the change of variables in Lemma~\ref{lem-TO}, one can design a general superlensing scheme in which one does not require $F$ (and also G) to be a Kelvin transform, and the lens is not required  to be radially symmetric, see \cite[Theorems 1 and 2 and Corollary 2]{Ng-Complementary} and \cite[Theorem 2]{Ng-Superlensing-Maxwell} for a discussion on the acoustic and electromagnetic settings, respectively. 
\end{remark}

\section{Cloaking using complementary media} \label{sect-CCM}

Cloaking using  complementary media was suggested by Lai et al.  \cite{LaiChenZhangChanComplementary}. The idea is to cancel the effect of an object by its complementary medium, a concept considered in \cite{PendryRamakrishna0}, see \cite{Ng-Complementary} for a discussion of this concept from  mathematical point of views. The study of cloaking using complementary media faces two difficulties.  Firstly, this problem is unstable since the equations describing the phenomenon  have sign changing coefficients, hence the ellipticity and the compactness are lost in general. Secondly,  localized resonance might appear, as shown in simulations in \cite{LaiChenZhangChanComplementary}. 

Cloaking using complementary media was mathematically justified for acoustic waves  \cite{Ng-Negative-Cloaking}  and for electromagnetic waves \cite{Ng-Cloaking-Maxwell}. The schemes that were used in \cite{Ng-Negative-Cloaking} and \cite{Ng-Cloaking-Maxwell} are inspired by the work of Lai. et al. and the study of complementary concept in \cite{Ng-Complementary,Ng-Superlensing-Maxwell}. Nevertheless, these schemes are different from the ones in \cite{LaiChenZhangChanComplementary}. The modification, mentioned below,  is necessary, as shown in the acoustic setting in \cite{Ng-CALR-object}; without the modification, cloaking might not be achieved (see also Section~\ref{sect-CALR-object}, Proposition~\ref{pro-CALR},  in particular,  and the comments following). 

Let us describe how to cloak the region $B_{2r_2} \setminus B_{r_2}$ for some $r_{2}> 0$ in the spirit of \cite{Ng-Negative-Cloaking}. We first consider the quasistatic regime. Assume that the cloaked region is  characterized by a matrix $a$,  which is  symmetric and uniformly elliptic in $B_{2 r_2} \setminus B_{r_2}$. The cloaking device consists of {\it two parts}. The first one, in $B_{r_2} \setminus B_{r_1}$,  makes use of reflecting complementary media to cancel the effect of the cloaked region, and the second  one in $B_{r_1}$,   fills the space that ``disappears" from the cancellation by the homogeneous medium. 
For the first part, we modify  the strategy in \cite{LaiChenZhangChanComplementary}. Instead of $B_{2r_2} \setminus B_{r_2}$, we consider $B_{r_3} \setminus B_{r_2}$ for some $r_3 > 0$ as the cloaked region in which the medium is given by the matrix 
\begin{equation*}
a_e = \left\{ \begin{array}{cl} a & \mbox{ in } B_{2 r_2} \setminus B_{r_2}, \\[6pt]
I & \mbox{ in } B_{r_3} \setminus B_{2 r_2}. 
\end{array} \right. 
\end{equation*} 
We assume that 
\begin{equation}\label{SCM-Lipschitz}
a_e \in C^1(\bar B_{r_3} \setminus B_{r_2}). 
\end{equation}
The complementary medium in $B_{r_2} \setminus B_{r_1}$ is given by 
\begin{equation*}
- \big(F^{-1}\big)_*a_e, 
\end{equation*}
where $F: B_{r_2} \setminus \bar B_{r_1} \to B_{r_3} \setminus \bar B_{r_2}$ is the Kelvin transform with respect to  $\partial B_{r_2}$.  Concerning the second part, the medium in $B_{r_1}$ is given by 
\begin{equation}\label{CCM-choice-A}
\Big(r_3^2/r_2^2 \Big)^{d-2} I,  
\end{equation}
which is also different from that suggested by Lai et al.  \cite{LaiChenZhangChanComplementary}. 
The reason for this choice is to ensure that 
\begin{equation}\label{CCM-GF}
G_{*} F_{*} A = I \mbox{ in } B_{r_{3}},  
\end{equation}
where $A$ is defined in \eqref{CCM-defA-ss} below. 
In two dimensions, the medium in $B_{r_1}$ is $I$, as used by Lai et al. \cite{LaiChenZhangChanComplementary}, while it is not $I$ in three dimensions.  With the loss, the medium is characterized by $A_\delta : = s_\delta A$,  where 
\begin{equation}\label{CCM-defA-ss}
A = \left\{ \begin{array}{cl} 
a_e & \mbox{ in } B_{r_3} \setminus B_{r_2}, \\[6pt]
F^{-1}_*a_e & \mbox{ in } B_{r_2} \setminus B_{r_1}, \\[6pt]
\Big(r_3^2/r_2^2 \Big)^{d-2} I & \mbox{ in } B_{r_1},\\[6pt]
I & \mbox{ otherwise}, 
\end{array} \right. 
\quad \mbox{ and } \quad 
 s_\delta = \left\{ \begin{array}{cl} -1 - i \delta  & \mbox{ in } B_{r_2} \setminus B_{r_1}, \\[6pt]
1 & \mbox{ otherwise}, 
\end{array} \right.  \mbox{ for } \delta \ge 0.
\end{equation}

Let  $\Omega$ be a smooth bounded open subset of  $\R^d$ with $B_{r_3} \subset \subset \Omega$, and let  $f \in L^2(\Omega)$. Denote $u_\delta, \, \hu \in H^1_0(\Omega)$, respectively, the unique solution of 
\begin{equation}\label{sys-CCM1}
\dive (A_\delta \nabla u_\delta) = f \mbox{ in } \Omega \quad \mbox{ and } \quad \Delta \hu = f \mbox{ in } \Omega. 
\end{equation}
The cloaking property of this scheme is given in the following theorem. 

\begin{theorem}\label{CCM-thm1} Let  $d=2, \, 3$, $0< \delta < 1$, and $f \in L^{2}(\Omega)$ with $\supp f \subset \Omega \setminus  B_{r_{3}}$. Let $u_\delta, u \in H^1_0(\Omega)$ be the uniques solutions defined by \eqref{sys-CCM1}. 
For any $0< \alpha < 1$, there exists $\ell > 0$, depending only on $r_2$, $\alpha$,  and the ellipticity and the Lipschitz constants of $a_e$  such that if $r_3 > \ell r_2$ then   
\begin{equation}\label{CCM-key-point-1}
\| u_{\delta}  \|_{H^1(\Omega)}  \le C \delta^{(\alpha - 1) / 2} \| f\|_{L^2(\Omega)}, 
\end{equation}
and
\begin{equation}\label{CCM-key-point-2}
\| u_{\delta} - \hu \|_{H^1(\Omega \setminus B_{r_3})}  \le C \delta^{\alpha} \| f\|_{L^2(\Omega)}, 
\end{equation}
for some positive constant $C$ independent of $\delta$ and $f$. In particular, we have 
\begin{equation}\label{CCM-key-point-2}
u_{\delta} \to \hu \mbox{ in } H^{1}(\Omega \setminus \bar B_{r_3}) \mbox{ as } \delta \to 0. 
\end{equation}
\end{theorem}

For an observer outside $B_{r_3}$, the medium in $B_{r_3}$ given by $A_\delta$ looks as the homogeneous one by \eqref{CCM-key-point-1} for small $\delta$: one has cloaking. 

\begin{proof} Set 
$$
\beta = (2 + \alpha)/3.
$$
We have, by Lemma~\ref{lem-stability},  
\begin{equation}\label{CCM-thm1-p1}
\| u_\delta\|_{H^1(\Omega)} \le C \D(f, \delta), 
\end{equation}
where 
\begin{equation}\label{def-D}
\D(f, \delta) :=  \frac{1}{\delta} \left| \Im \int_{\Omega} f \bar u_\delta \right| + \| f \|_{L^2(\Omega)}. 
\end{equation}
As in the proof of Theorem~\ref{SCM-thm1}, define $u_{1, \delta} \in H^1_{\loc}(\R^d \setminus B_{r_2})$ and $u_{2, \delta} \in H^1(B_{r_3})$  as follows 
\begin{equation*}
u_{1, \delta} = u_\delta \circ F^{-1}  \mbox{ in } \R^d \setminus B_{r_2} \quad \mbox{ and } \quad u_{2, \delta} = u_{1, \delta} \circ G^{-1} = u_\delta \circ F^{-1} \circ G^{-1} \mbox{ in } B_{r_3}.  
\end{equation*}
We have, by Lemma~\ref{lem-TO},  
$$
\dive (A \nabla u_{1, \delta}) = 0 \mbox{ in } B_{r_{3}} \setminus B_{2 r_2}, \quad u_{1, \delta} = u_{\delta} \mbox{ on } \partial B_{r_2}, \quad \mbox{ and } \quad (1 + i \delta)A \nabla u_{1, \delta}  = A \nabla u_\delta |_{ext} \mbox{ on } \partial B_{r_2}. 
$$
Let ${\mathcal A}$ be a Lipschitz extension of $a_e$ in $B_{r_3}$ such that ${\mathcal A}(0) = I$ and let $w_\delta \in H^1_0(B_{r_3})$ be such that 
$$
\dive ({\mathcal A} \nabla w_\delta) = 0 \mbox{ in } B_{r_3} \setminus \partial B_{r_2} \quad \mbox{ and } \quad [\mathcal A \nabla w_\delta \cdot \nu]  = i \delta A \nabla u_{1, \delta} \mbox{ on } \partial B_{r_2}.  
$$
Then 
\begin{equation}\label{CCM-p3}
\| w_\delta \|_{H^1(B_{r_3})} \le C \delta \D(f, \delta)^{1/2}.  
\end{equation}
Applying a three-sphere inequality  \cite[Lemma 1]{Ng-Negative-Cloaking} to $(u_{1, \delta} - u _\delta) \chi_{B_{r_3} \setminus B_{r_2}} - w_\delta$ in $B_{r_3}$ \footnote{Recall that $\chi_D$ denotes the characteristic function of a subset $D$ of $\R^d$.} and using \eqref{CCM-p3},  we obtain, if $\ell$ is sufficiently large, that
\begin{equation}\label{CCM-thm1-p2}
\|u_\delta - u_{1, \delta} \|_{H^{1/2}(\partial B_{2 r_2})} + \| \partial_r (u_\delta - u_{1, \delta})|_{ext} \|_{H^{-1/2}(\partial B_{2 r_2})} \le C \delta^\beta \D(f, \delta)^{1/2}.  
\end{equation}
In the spirit of  \eqref{SCM-hu-delta}, we define
\begin{equation}\label{CCM-hu-delta}
\hu_\delta = \left\{\begin{array}{cl} u_\delta & \mbox{ in } \Omega \setminus B_{r_3}, \\[6pt]
u_\delta - (u_{1, \delta} - u_{2, \delta}) & \mbox{ in } B_{r_3} \setminus B_{2 r_2}, \\[6pt]
u_{2, \delta} & \mbox{ in } B_{2 r_2}. 
\end{array}\right. 
\end{equation}
We have 
\begin{equation*}
\Delta (\hu_\delta - \hu)  = 0 \mbox{ in } \Omega \setminus (\partial B_{r_3} \cup \partial B_{2 r_2}),  
\end{equation*}
\begin{equation*}
[\hu_\delta - \hu] = 0 \mbox{ on } \partial B_{r_3}, \quad [\partial_r (\hu_\delta - \hu)] = - i \delta \partial_r u_{1, \delta}|_{int} \mbox{ on } \partial B_{r_3},
\end{equation*}
and
\begin{equation*}
[\hu_\delta - \hu] = u_{\delta} - u_{1, \delta},  \quad [\partial_r (u_\delta - \hu) ] = \partial_r (u_{\delta}|_{ext} - u_{1, \delta})   \quad \mbox{ on } \partial B_{2 r_2}.  
\end{equation*}
By Lemma~\ref{lem-stability}, we obtain from \eqref{CCM-thm1-p1} and \eqref{CCM-thm1-p2} that 
\begin{equation}\label{CCM-p0}
\| \hu_\delta - \hu_0 \|_{H^1(\Omega)} \le C \delta^{\beta} \D(f, \delta)^{1/2}. 
\end{equation}
By \eqref{def-D}, this implies, since $\beta > 1/2$, that  
\begin{equation}\label{CCM-p1}
\| \hu_\delta \|_{H^1(\Omega \setminus B_{r_3})} \le C \|f \|_{L^2(\Omega)}. 
\end{equation}
We derive from \eqref{def-D} and \eqref{CCM-p1} that 
\begin{equation}\label{CCM-p2}
\D(f, \delta) \le C \delta^{-1} \| f\|_{L^2(\Omega)}^2
\end{equation}
and from \eqref{CCM-p0} and \eqref{CCM-p1} that 
\begin{equation}\label{CCM-p3-1}
\| \hu_\delta - \hu \|_{H^1(\Omega)} \le C \delta^{\beta - 1/2} \| f\|_{L^2(\Omega)}. 
\end{equation}

Up to this point, the analysis is in the spirit of  \cite{Ng-Negative-Cloaking}, and now we add some new ingredients to derive the desired conclusions. We have, by \eqref{CCM-p3-1},  
\begin{equation*}
\left|\int_{\Omega} f \bar \hu_\delta - \int_{\Omega} f \bar \hu \right| \le C \delta^{\beta - 1/2} \| f\|_{L^2(\Omega)}^2
\end{equation*}
and, by multiplying the equation of $\hu$ with $\bar \hu$ and considering the imaginary part,  
\begin{equation*}
\Im \int_{\Omega} f \bar \hu = 0. 
\end{equation*}
It follows from \eqref{def-D} that 
\begin{equation*}
\D(f, \delta) \le C \delta^{\beta - 3/2} \| f\|_{L^2(\Omega)}^2. 
\end{equation*}
From \eqref{CCM-p0}, we obtain 
\begin{equation*}
\| \hu_\delta - \hu \|_{H^1(\Omega)} \le C \delta^{3\beta/2 - 3/4} \| f\|_{L^2(\Omega)}. 
\end{equation*}
Repeating this process, one reaches, for $n \ge 1$,  that
\begin{equation*}
\D(f, \delta) \le C_n \delta^{\beta(1 + 1/2 + .. + 1/2^{n-1})  - (1/2 + .. + 1/2^{n}) - 1} \| f\|_{L^2(\Omega)}^2
\end{equation*}
and
\begin{equation*}
\| \hu_\delta - \hu \|_{H^1(\Omega)} \le C_{n} \delta^{\beta(1 + 1/2 + .. + 1/2^n)  - (1/2 + .. + 1/2^{n+1})} \| f\|_{L^2(\Omega)}, 
\end{equation*}
where $C_n$ is a positive constant independent of $\delta$ and $f$. 
The conclusion follows by taking  $n$ large enough. 
\end{proof}

\begin{remark} \rm One of the crucial steps of this proof is to introduce the function $\hu$. In general $u_{1, \delta} - u_\delta $ explodes in the region $B_{r_3}\setminus B_{2 r_2}$. A numerical simulation of this fact is given in the work of Lai. et al. \cite{LaiChenZhangChanComplementary}.  A mathematical illustration of this phenomenon can be seen from the explicit representation of $u_\delta - u_{1, \delta}$ in $B_{r_3} \setminus B_{2 r_2}$ using separation of variables, see \cite[Proof of Theorem 1]{Ng-Negative-Cloaking}.  The definition of $\hu$ is inspired by the concept of the normalizing energy used in the study of the Ginzburg-Landau equation, see, e.g.,  \cite{BBH}.  
\end{remark}

We next present the result in the finite frequency regime. 
Assume that the cloaked region is  characterized by a matrix $a$  that is  symmetric, uniformly elliptic and  a bounded complex function $\sigma$ that  satisfies $\Re \sigma > c > 0$  and $\Im \sigma \ge 0$ both defined in $B_{2 r_2} \setminus B_{r_2}$. 
As in the spirit of the zero-frequency case, we consider the layer $B_{r_3} \setminus B_{r_2}$ as the cloaked region that is characterized by 
$$
a_e, \sigma_e =  \left\{ \begin{array}{cl} a, \sigma & \mbox{ in } B_{2 r_2} \setminus B_{r_2}, \\[6pt]
I, 1 & \mbox{ in } B_{r_3} \setminus B_{2 r_2}. 
\end{array} \right. 
$$
The cloaking device consists of {\it two parts}. The first one, the complementary layer in $B_{r_2} \setminus B_{r_1}$, is characterized by 
\begin{equation*}
- \big(F^{-1}\big)_*a_e, - F^{-1}_*\sigma_e. 
\end{equation*}
Concerning the second part, the medium in $B_{r_1}$ is given by 
\begin{equation}\label{CCM-choice-A}
\Big(r_3^2/r_2^2 \Big)^{d-2} I, (r_3^2/r_2^2)^d.   
\end{equation}
Again, the reason for this choice is to ensure
\begin{equation}\label{CCM-GF}
G_{*} F_{*} A= I \quad \mbox{ and } \quad ,  G_{*} F_{*} \Sigma = 1 \mbox{ in } B_{r_{3}}, 
\end{equation}
where $A$ and $\Sigma$ are defined in \eqref{CCM-defAS-ss}.
We will assume that \eqref{SCM-Lipschitz} holds. 
Set  $A_\delta : = s_\delta A$  and $\Sigma_\delta : = s_\delta \Sigma$, where, for $\delta \ge 0$,  
\begin{equation}\label{CCM-defAS-ss}
A, \Sigma = \left\{ \begin{array}{cl} 
a_e, \sigma_e & \mbox{ in } B_{r_3} \setminus B_{r_2}, \\[6pt]
F^{-1}_*a_e, F^{-1}_*\sigma_e & \mbox{ in } B_{r_2} \setminus B_{r_1}, \\[6pt]
\Big(r_3^2/r_2^2 \Big)^{d-2} I,  (r_3^2/r_2^2)^d & \mbox{ in } B_{r_1},\\[6pt]
I, 1 & \mbox{ otherwise}, 
\end{array} \right. 
\quad \mbox{ and } \quad 
 s_\delta = \left\{ \begin{array}{cl} -1 - i \delta  & \mbox{ in } B_{r_2} \setminus B_{r_1}, \\[6pt]
1 & \mbox{ otherwise}. 
\end{array} \right. 
\end{equation}
Let  $k>0$,  $f \in L^2(\R^d)$ with compact support and  denote $u_\delta, \hu \in H^1_{\loc}(\R^d)$, respectively, the unique outgoing solutions of 
\begin{equation}\label{sys-CCM1}
\dive (A_\delta \nabla u_\delta) + k^2 \Sigma_\delta u_\delta = f \mbox{ in } \R^d \quad \mbox{ and } \quad \Delta \hu  + k^2 \hu = f \mbox{ in } \R^d. 
\end{equation}
Here is the variant of Theorem~\ref{CCM-thm1} for the finite frequency regime, which confirms the cloaking property of the scheme considered. 

\begin{theorem}\label{CCM-thm2} Let  $d=2, \, 3$, $k>0$, $0< \delta < 1$,  $R_0  > r_3$, $f \in L^{2}(\R^d)$ with $\supp f \subset B_{R_0} \setminus  B_{r_{3}}$. Let $u_\delta, u \in H^1_{\loc}(\R^d)$ be the unique outgoing solutions defined by \eqref{sys-CCM1}. 
For any $0< \alpha < 1$, there exists $\ell > 0$, depending only on $r_2$, $\alpha$,  and the ellipticity and the Lipschitz constants of $a_e$  such that if $r_3 > \ell r_2$ then   
\begin{equation}\label{CCM-key-point-1-2}
\| u_{\delta}  \|_{H^1(B_R)}  \le C \delta^{(\alpha - 1) / 2} \| f\|_{L^2(\R^d)}, 
\end{equation}
and
\begin{equation}\label{CCM-key-point-2-2}
\| u_{\delta} - \hu \|_{H^1(B_R \setminus B_{r_3})}  \le C \delta^{\alpha} \| f\|_{L^2(\R^d)}, 
\end{equation}
for some positive constant $C$ independent of $\delta$ and $f$. In particular, we have 
\begin{equation}\label{CCM-key-point-2-2}
u_{\delta} \to \hu \mbox{ in } H^{1}_{\loc} (\R^d \setminus \bar B_{r_3}) \mbox{ as } \delta \to 0. 
\end{equation}
\end{theorem}

\begin{proof} The proof of Theorem~\ref{CCM-thm2} is in the spirit of Theorem~\ref{CCM-thm1} with a crucial point being the establishment of  \eqref{CCM-thm1-p2} in the finite frequency regime. This can be done as follows. On one hand, we  have, by  \cite[Theorem 2]{MinhLoc2}, 
\begin{equation}\label{SCM-thm2-p1}
\|u_\delta - u_{1, \delta} \|_{\bH(\partial B_{2 r_2})} 
\le C \| \partial_r (u_\delta - u_{1, \delta})|_{ext} \|_{\bH(\partial B_{r_2})}^\alpha 
 \Big( \|u_\delta - u_{1, \delta} \|_{\bH(\partial B_{4 r_2})} + \|u_\delta |_{ext} - u_{1, \delta}  \|_{\bH(\partial B_{r_2})}  \Big)^{1-\alpha}.  
\end{equation}
for some positive constant $\tau$ depending only on $r_2$ and the ellipticity and the Lipschitz of $a_e$. 
Here we denote
$$
\|v \|_{\bH(\partial B_{r})} :=  \|v \|_{H^{1/2}(\partial B_{r})} + \|A \nabla v \cdot \nu \|_{H^{-1/2}(\partial B_{r})}.
$$
On the other hand, we obtain, by \cite[Lemma 6]{Ng-Cloaking-Maxwell}, that 
\begin{equation}\label{SCM-thm2-p2}
\|u_\delta - u_{1, \delta} \|_{\bH(\partial B_{4 r_2})} \le C_k \|u_\delta - u_{1, \delta} \|_{\bH(\partial B_{2 r_2})}^{\xi} \|u_\delta - u_{1, \delta} \|_{\bH(\partial B_{r_3})}^{1 - \xi}, 
\end{equation}
where $\xi = \ln \big(r_3/ (4 r_2) \big)/ \ln \big(r_3/ (2 r_2) \big)$. 
Combining \eqref{SCM-thm2-p1} and \eqref{SCM-thm2-p2} yields \eqref{CCM-thm1-p2} if $\ell$ is sufficiently large. The rest of the proof is in the spirit of Theorem~\ref{SCM-thm1}. The details are omitted. 
\end{proof}

\begin{remark} \rm Previous given proof of cloaking using complementary media \cite{Ng-Negative-Cloaking} can be extended to the finite frequency regime. Nevertheless, the size of the cloaked object (the cloaked region) is small as $k$ is large. In \cite{MinhLoc2}, we extended the approach in \cite{Ng-Negative-Cloaking} for the finite frequency regime in which  the size of the object can be independent of the frequency $k$. In fact, we showed that there exists $\lambda_0 > 1$ depending on the ellipticity and the Lipschitz of $a_e$ such that one can cloak an object inside $B_{\lambda_0 r_2} \setminus B_{r_2}$; nevertheless $\lambda_0$ can be smaller than $2$ but one can choose a large $r_2$  to compensate this.  The proof given here is again in the spirit of the work \cite{Ng-Cloaking-Maxwell} in which cloaking using complementary media for electromagnetic waves is investigated.  
\end{remark}

\begin{remark} \rm Using the change of variables in Lemma~\ref{lem-TO}, one can design a general cloaking scheme in which one does not requires $F$ (and also G) to be a Kelvin transform and the cloaking device is not necessary to be radially symmetric, see \cite{Ng-Cloaking-Maxwell} for a discussion in the electromagnetic setting.  
\end{remark}

\section{Cloaking an object via anomalous localized resonance}\label{sect-CALR-object}

In this section, we present another cloaking technique using NIMs namely cloaking {\it an object} via anomalous localized resonance. The advantage of this cloaking technique over the one using complementary media is that the cloaking devices used here are independent of the cloaked object. This cloaking technique was suggested in \cite{Ng-CALR-object} and inspired from \cite{MiltonNicorovici, Ng-CALR}.
In the two dimensional case, it was shown in \cite{Ng-CALR-object} that the negative index layer of the  lens considered in Section~\ref{sect-SCM} can act like a cloaking device for a finite size object near by, see Figure~\ref{fig-CALR}.  More precisely, in the quasistatic regime, we have 

\begin{theorem}\label{thm-CALR} Let $d=2$,  $0 <  r_0 <  r_1 < r_2 $, $x_1 \in \partial B_{r_1}$, and $x_2  \in \partial B_{r_2}$. Set  $r_3 := r_2^2/ r_1$ and  ${\mathcal C} := \big(B(x_1, r_0) \cap B_{r_1} \big) \cup  \big(B(x_2, r_0) \cap (B_{r_3}  \setminus B_{r_2}) \big)$, assume that $B_{r_3} \subset \subset \Omega$  and let $a_c$ be a symmetric uniformly elliptic matrix-valued function defined in ${\mathcal C}$. 
Define 
\begin{equation}\label{def-AS-CS}
A_c = \left\{\begin{array}{cl} 
a_c &  \mbox{ in } {\mathcal C}, \\[6pt]
I & \mbox{ otherwise},
\end{array}\right. \quad 
\mbox{ and } \quad 
s_\delta = \left\{\begin{array}{cl}  -1 -  i \delta &  \mbox{ in } B_{r_2} \setminus B_{r_1}, \\[6pt]
1 & \mbox{ otherwise}.  
\end{array}\right. 
\end{equation} 
Given $f \in L^2(\Omega)$ with $\supp f \subset \Omega \setminus B_{r_3}$,  let  $u_\delta, \, \hu \in H^1_0(\Omega)$, respectively, be the unique solution to the equations
\begin{equation}
\dive(A_\delta \nabla u_\delta )    = f \mbox{ in } \Omega  \quad \mbox{ and } \quad 
\Delta \hu     = f \mbox{ in } \Omega. 
\end{equation}
For any $0 < \alpha < 1$, there exists $r_0(\alpha) > 0$ that depends only on $\alpha$, $r_1$, and $r_2$,  such that  if $r_0 < r_0(\alpha)$ then
\begin{equation}\label{result-CS}
\| u_\delta  -  \hu \|_{H^1(\Omega \setminus B_{r_3})} \le C \delta^{\alpha} \| f\|_{L^2(\Omega)}, 
\end{equation}
where $C$ is a positive constant independent of $f$, $\delta$, $r_0$, $x_1$, and $x_2$. 
\end{theorem}

\begin{figure}
\centering
\begin{tikzpicture}[scale=1]

\draw[fill = blue] ({cos(240)},{sin(240)}) circle (0.2cm);

\draw[fill = blue] ({1.8*cos(150)}, {1.8*sin(150)}) circle (0.2cm);

\draw[line width = 0.8cm, color = red] (0,0) circle (1.4cm);

\draw (0,0) circle (1.8cm);

\draw (0,0) circle (1.cm);


\draw[->] (0,0) -- (1,0);
\draw (0.7, 0) node[below]{$r_1$};

\draw[->] (0,0) -- ({1.8*cos(60)}, {1.8*sin(60)});
\draw ({1.2*cos(60)}, {1.2*sin(60)}) node[right]{$r_2$};

\draw (-0.5, 0) node[]{$I$};

\draw (0., -1.5) node[]{$-I$};

\draw (-2.5, 0) node[]{$I$};

\draw (4, 0) node[]{$\Longleftrightarrow$};
 
\draw (6.2, 0) node[]{$I$};

\end{tikzpicture}
\caption{The red lens layer will cloak the blue region ${\mathcal C}$. An observer outside $B_{r_3}$ $(r_3 = r_2^2/ r_1)$ sees neither the red layer nor the blue regions.}
\label{fig-CALR}
\end{figure}
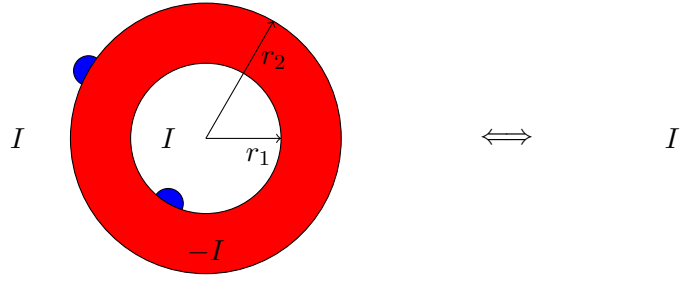

\begin{proof}
Set 
$$
\beta = (2 + \alpha)/3. 
$$
We have, by Lemma~\ref{lem-stability},  
\begin{equation}\label{part0-thm-CS}
\| u_\delta\|_{H^1(\Omega)} \le C \D(f, \delta), 
\end{equation}
where 
\begin{equation}\label{def-D-CS}
\D(f, \delta) :=  \frac{1}{\delta} \left| \Im \int_{\Omega} f \bar u_\delta \right| + \| f \|_{L^2(\Omega)}^2. 
\end{equation}
As in the proof of Theorem~\ref{SCM-thm1}, define $u_{1, \delta} \in H^1_{\loc}(\R^2 \setminus B_{r_2})$ and $u_{2, \delta} \in H^1(B_{r_3})$  as follows 
\begin{equation*}
u_{1, \delta} = u_\delta \circ F^{-1}  \mbox{ in } \R^2 \setminus B_{r_2} \quad \mbox{ and } \quad u_{2, \delta} = u_{1, \delta} \circ G^{-1}  \mbox{ in } B_{r_3}.  
\end{equation*}

Set 
$$
S = \big( B_{r_3} \setminus B_{r_2} \big)  \cap \big(B(x_2, r_0) \cup G \circ F (B(x_1, r_0) \cap B_{r_1}) \big).  
$$
By Lemma~\ref{lem-TO}, we have 
\begin{equation}\label{transmission1-pro2}
u_{1, \delta} -  u_\delta = 0  \quad \mbox{ and }  \quad \partial_r u_{1, \delta} |_{ext} - \partial_r u_{\delta}  = - i \delta \partial_r u_{1, \delta} |_{ext} \mbox{ on } \partial B_{r_2} \setminus \partial S,  
\end{equation}
and 
\begin{equation}\label{transmission2-pro2}
u_{2, \delta} - u_{1, \delta}  = 0 \quad  \mbox{ and }  \quad  \partial_r u_{2, \delta} |_{int}  - \partial_r u_{1, \delta} = i \delta \partial_r u_{1, \delta} |_{int} \mbox{ on } \partial  B_{r_3} \setminus \partial S. 
\end{equation}
Applying Lemma~\ref{lem-TO}, we obtain
$$
\Delta  u_{1, \delta}  = 0 \mbox{ in } B_{r_3} \setminus B_{r_2}
$$
and 
$$
\Delta u_{2, \delta}  = 0 \mbox{ in } B_{r_3}  \setminus \big( (G \circ F) (B(x_1, r_0) \cap B_{r_1}) \big).  
$$
Recall that 
\begin{equation}\label{haha121}
\Delta u_{\delta} = 0 \mbox{ in } (B_{r_3} \setminus B_{r_2} ) \setminus {\mathcal C}. 
\end{equation}
Denote $x_3 \in \partial B_{r_3}$  the image of $x_1$ by $F$. 
The new key ingredient in comparison with the approach used in the proof of Theorem~\ref{CCM-thm1} is the fact that there exist two constants $R_2, R_3 > 0$ such that if $r_0$ is small enough and if one defines 
$$
O_{2} = B_{r_2} \cup \{|z - x_2| <  R_2 \}, \quad O_{3} = B_{r_3} \setminus \{|z - x_3| <  R_3 \}, \quad \mbox{ and } \quad O = O_3 \setminus O_2,  
$$
then 
\begin{equation} \label{claim1-pro2}
\| u_{1,\delta} - u_{\delta}\|_{H^{1/2}(\partial O_{2})} + \| \partial_r (u_{1,\delta} - u_{\delta}) \|_{H^{-1/2}(\partial O_{2})}  \le C \delta^{\beta} \| u_{\delta}\|_{H^1(B_{r_3})}
\end{equation}
and
\begin{equation} \label{claim2-pro2}
\| u_{2,\delta} - u_{1, \delta}\|_{H^{1/2}(\partial O_{3})} + \| \partial_r (u_{2,\delta} - u_{1, \delta}) \|_{H^{-1/2}(\partial O_{3})}  \le C \delta^{\beta} \| u_{\delta}\|_{H^1(B_{r_3})}.  
\end{equation}
The details of the proof of this fact, which are out of the scope of this survey,  are given in \cite[Section 3.1]{Ng-CALR-object}. 
Define 
\begin{equation}\label{def-Udelta-1}
\hu_\delta  = \left\{\begin{array}{cl} u_\delta & \mbox{ in } \Omega \setminus O_3, \\[6pt]
u_{2, \delta} -   (u_{1, \delta} - u_{\delta}) & \mbox{ in } O, \\[6pt]
u_{2, \delta} &  \mbox{ in } O_2. 
\end{array}\right. 
\end{equation}
Then, $\hu_\delta \in H^1 \big(\Omega \setminus \partial O \big)$ with $\hu_\delta = 0$ on $\partial \Omega$ is a solution of the equation
\begin{equation*}
\Delta \hu_\delta = f \mbox{ in } \Omega \setminus  \partial O. 
\end{equation*}
This implies, by \eqref{claim1-pro2} and \eqref{claim2-pro2}, 
\begin{equation}\label{part0-pro2}
\|\hu_\delta - \hu \|_{H^1(\Omega \setminus \partial O)}
\le C \delta^{\beta} \| u_\delta \|_{H^1(B_{r_3})}. 
\end{equation}
Since $\beta > 1/2$,  it follows from \eqref{def-D-CS}  that 
\begin{equation}\label{t1}
\| \hu_\delta\|_{H^1(\Omega \setminus \partial O)} \le  C  \|f \|_{L^2(\Omega)}. 
\end{equation}
This in turn implies that 
\begin{equation}\label{t1}
\D(f, \delta) \le  C \delta^{- 1}  \|f \|_{L^2(\Omega)}^2
\end{equation}
and
\begin{equation}\label{ud-improve2}
\|\hu_\delta - \hu \|_{H^1(\Omega \setminus \partial O)}
\le C \delta^{\beta - 1/ 2} \| f\|_{L^2(\Omega)}. 
\end{equation}
Involving  the arguments used in the last part of the proof of Theorem~\ref{CCM-thm1}, we have, for $n \ge 1$, 
\begin{equation*}
\D(f, \delta) \le  C \delta^{\beta (1 + .. + 1/2^{n-1}) - (1 + .. + 1/2^n) - 1}  \|f \|_{L^2(\Omega)}^2
\end{equation*}
and
\begin{equation*} 
\|\hu_\delta - \hu \|_{H^1(\Omega \setminus \partial O)}
\le C \delta^{\beta (1 + .. + 1/2^n) - (1/ 2 + .. + 1/2^{n+1})} \| f\|_{L^2(\Omega)}. 
\end{equation*}
The conclusion follows by taking $n$ sufficiently large. 
\end{proof}

\begin{remark} \rm As mentioned, one of the key ingredients are \eqref{claim1-pro2} and \eqref{claim2-pro2}. This is based on a three-sphere inequality with a partial information, see \cite[Section 3.1]{Ng-CALR}. The proof of this result also involves the properties of conformal maps. A variant of these inequalities holds for the Helmholtz equation in two dimensions. Due to the use of the conformal maps in two dimensions, we do not know if the variants of \eqref{claim1-pro2} and \eqref{claim2-pro2} hold for three dimensions. Nevertheless, a modification of the cloaking construction can be made to obtain a cloaking device that can cloak some  finite region near by. The modification is based on the concept of doubly complementary media that was first introduced in \cite{Ng-CALR} with its roots in \cite{Ng-Complementary}.  The interested reader can find a detailed discussion in  \cite{Ng-CALR-object}.  
\end{remark}

Invoking   ideas similar to those in the proof of Theorem~\ref{thm-CALR}, we establish,  see \cite[Proposition 3.2]{Ng-CALR-object}, that  

\begin{proposition}\label{pro-CALR} Let $d=2$, $0< \delta < 1$,  $0 < r_0 <  r_1 < r_2 $,  and $x_3 \in \partial B_{r_3}$ with $r_3 = r_2^2/ r_1$. Assume that $B_{r_3} \subset \subset \Omega$ and let $f \in L^{2}(\Omega)$ with $\supp f \subset \Omega \setminus  B_{r_{3}}$. 
Let $a_c$ be a symmetric uniformly elliptic matrix-valued function defined in $B(x_3, r_0) \cap B_{r_3}$.
Let $u_\delta \in H^1_0(\Omega)$ be the unique solution of \eqref{sys-CCM1} where 
\begin{equation*}
\dive(s_\delta A \nabla u_\delta )  = f \mbox{ in } \Omega. 
\end{equation*}
Here $(A, \Sigma)$ is given by \eqref{CCM-defA-ss} where 
\begin{equation*}
a_e = \left\{\begin{array}{cl} a_c &  \mbox{ in } B(x_3, r_0) \cap B_{r_3}, \\[6pt]
I & \mbox{ in } (B_{r_3} \setminus B_{r_2}) \setminus B(x_3, r_0).  
\end{array}\right.
\end{equation*}
There exists $r_* > 0$ depending only on $r_1$ and $r_2$ such that if $r_0 < r_*$, then 
\begin{equation}\label{tt1}
u_\delta  \to \cU \mbox{ in } L^2(\Omega \setminus B_{r_3}). 
\end{equation}
Here $\cU \in H^1_{0}(\Omega)$ is the unique  outgoing solution to the equation  
\begin{equation}\label{tt2}
\dive(\cA \nabla \cU)  = f \mbox{ in } \Omega, \mbox{ where } \cA = \left\{\begin{array}{cl} a_c& \mbox{ in } B(x_0, r_0) \cap B_{r_3}, \\[6pt]
I& \mbox{ otherwise}. 
\end{array}\right.
\end{equation}
\end{proposition}

From \eqref{tt1} and \eqref{tt2}, one concludes that the object in $B_{r_3} \setminus B_{r_2}$ is not cloaked by its complementary medium in $B_{r_2} \setminus B_{r_1}$ as suggested in  \cite{LaiChenZhangChanComplementary}  and as is  usually accepted in the literature.

\section{Electromagnetic wave propagation in media consisting of dispersive metamaterials} \label{sect-Maxwell}

The fundamental Maxwell's equations -- without source --  are
\begin{equation}\label{eq:EM-G}
\left\{ \begin{aligned}
& \partial_t D(t, x) = \nabla \times   H(t, x), \\
& \partial_t  B(t, x) =- \nabla \times  E(t, x),
\end{aligned} \right. \quad  \mbox{ for } t \in \R,\ x \in \R^3,
\end{equation}
where $ E \in \R^3$ (resp. $ H \in \R^3$) is the electric (resp. magnetic) field and $ D\in \R^3$ (resp. $ B\in \R^3$) is the electric (resp. magnetic) induction field.  In order to close the system \eqref{eq:EM-G}, one adds constitutive relations that express $D$ and $B$ as functions of $E$ and $H$. For dispersive media, these relations are frequency dependent.  Taking these constitutive relations into account, the corresponding system of \eqref{eq:EM-G} in the time domain has the form   
\begin{equation}\label{eq:EM-G-1}
\left\{ \begin{aligned}
& \eps_{rel}(x) \partial_t E(t, x) + (\lambda_{ee} * E)(t, x) +  (\lambda_{em} *H)(t, x) =  \nabla \times H(t, x), \\
&  \mu_{rel}(x) \partial_t H(t, x)+  (\lambda_{me}* E)(t, x) + (\lambda_{mm}*H)(t,x)   = -\nabla \times E(t, x),
\end{aligned}\right.
 \quad  t \in \R,\  x \in \R^3,
\end{equation}
where $*$ stands for the convolution with respect to time $t$. Here the following conventions/assumptions are imposed on $\eps_{rel}$, $\mu_{rel}$, and  $\lambda_{ij}$ for $i, j \in \{e, m \}$: 
\begin{equation}\label{epsmu}
\mbox{$\eps_{rel}$ and $\mu_{rel}$ are two  $3 \times 3$ real symmetric uniformly elliptic matrices defined in $\R^3$.}
\end{equation}
and
\begin{equation}\label{assumption-lambda}
\widehat{\lambda_{i j}},  \, \lambda_{i j} \in L^1_{\loc}\big(\R, L^\infty( \R^3)^{3 \times 3} \big), \ \mbox{ and } \ \lambda_{ij}  \mbox{ is real-valued},  \qquad \mbox{ for } (i, j) \in \big\{e, m \big\}^2. 
\end{equation}
 In this section, for a time-dependent quantity $X(t, x)$, its temporal Fourier transform  is given by
\begin{equation}\label{FT}
\widehat{ X}(\omega, x):= \frac{1}{\sqrt{2 \pi}}\int_{\R} X(t, x) e^{i\omega t}\, d t,\quad  \mbox{ for } \omega \in \R,\  x \in \R^3.
\end{equation}
Let  $\chi_{ij}$ be the susceptibilities that characterizes the dispersive effects of the medium. The connection between $\lambda_{ij}$ and $\chi_{ij}$ is 
\begin{equation}\label{def-lambda}
\widehat{\lambda_{ij}}(\omega, x) := - i \omega \widehat{\chi_{ij}}(\omega, x), \quad \mbox{ for } (i,j) \in \{e,m\}^2,\ \omega \in \R,\  x \in \R^3. 
\end{equation}
 The permittivity $\eps$ and the permeability $\mu$ of the medium are given  by
\begin{equation}
\widehat{\eps} : = \eps_{rel} + \widehat \chi_{ee} \quad \mbox{ and } \quad \widehat{\mu}: = \mu_{rel} + \widehat \chi_{mm}.
\end{equation}
The details of deriving \eqref{eq:EM-G-1} from \eqref{eq:EM-G} using the appropriate assumptions on dispersive media are given in \cite[Section 2]{Ng-Vinoles}. 

\medskip 
Two fundamental assumptions physically relevant to the model,  causality and passivity,  are imposed.  

\medskip 
\noindent {\bf Causality}: the effect
cannot precede the cause, i.e., the present states of the system depend only on its states in the past. Mathematically, one requires 
\begin{equation}\label{eq:causality}
\lambda_{ij}(t) = 0, \qquad \text{for all $t<0$ and for all $(i,j) \in \{e,m\}^2$}.
\end{equation}

Under this assumption, we have, for $(i, j)  \in \{e, m \}^2$, 
\begin{equation}\label{eq-causality1}
(\lambda_{ij} *  X)(t,\cdot) = \int_{-\infty}^t \lambda (t-\tau,\cdot) X(\tau,\cdot)\,  d \tau = \int_{0}^\infty \lambda (\tau,\cdot) X(t - \tau,\cdot)\,  d \tau,\quad  \mbox{ for } t \in \R. 
\end{equation}



\medskip 
\noindent {\bf Passivity}: One assumes,  for almost every $x \in \R^3$, for almost every $\omega \in \R$, and for all $X \in \mathbb{C}^6$ \footnote{Here $\C$ denotes the set of complex numbers.}, that\footnote{Here $\cdot$ stands for the Euclidean scalar product in $\C^6$.}
\begin{equation}\label{eq:passivity}
\real \left( \begin{bmatrix}
\widehat{\lambda_{ee}}(\omega,x) & \widehat{\lambda_{em}}(\omega,x) \\ 
\widehat{\lambda_{me}}(\omega,x) & \widehat{\lambda_{mm}}(\omega,x)
\end{bmatrix} X \cdot \overline X \right) \ge 0, 
\end{equation}
Assumption \eqref{eq:passivity} means that the medium is dissipative, i.e., it 
does not produce electromagnetic energy by itself.  



In the anisotropic case $(\chi_{em} = \chi_{me} = 0)$, condition~\eqref{eq:passivity} is equivalent to\footnote{Here for a $3 \times 3$ matrix $A$, we denote $A \le 0$ if $A x \cdot x \le 0$ for all $x \in \R^3$.}
\begin{equation}\label{sign-eps-mu}
\omega \imag \widehat \eps(\omega), \ \omega \imag \widehat \mu(\omega)  \ge 0,\quad  \mbox{ for almost all } \omega \in \R. 
\end{equation}
Condition~\eqref{sign-eps-mu} ensures that when small loss is added, the problem associated with the outgoing (Silver-M\"uller) condition at infinity is well-posed (see,  e.g., \cite{Ng-Superlensing-Maxwell}). Adding a small loss is the standard mechanism to study phenomena related to metamaterials in the frequency domain. 
Nevertheless, condition~\eqref{sign-eps-mu} does not exclude the ill-posedness in the frequency domain when the loss is 0 (see \cite[Proposition 2]{Ng-WP}). As one sees later, even if the problem is ill-posed in the frequency domain for some frequency, the well-posedness is roughly ensured for the problem in the time domain under   the causality and passivity conditions mentioned above (see Theorem~\ref{thm-WP}). 




\medskip 
One of typical classes of dispersive anisotropic media ($\chi_{me}  = \chi_{em} = 0$) satisfying condition \eqref{assumption-lambda},  the causality \eqref{eq:causality} and the passivity \eqref{eq:passivity} 
is the class of  media obeying \emph{Lorentz' model}. For a homogeneous isotropic medium, the susceptibilities $\chi_{ee}$ and $\chi_{mm}$ are of the form (see e.g.,  \cite[(7.51)]{jackson})
\begin{equation}\label{eq:LorentzModelFrequencyDomain}
\widehat{\chi}(\omega) = \sum_{\ell=1}^n \frac{\omega_{p,\ell}^2}{\omega_{0,\ell}^2 - \omega^2  - 2i \gamma_\ell \omega}\, I,\quad  \mbox{ for } \omega \in \R, 
\end{equation}
where $\omega_{p,\ell}$ (resp. $\omega_{0,\ell}$ and $\gamma_\ell$) are positive (resp. non-negative)  material constants (recall that $I$ is the identity matrix). Using the residue theorem, one can  show (see e.g.,  \cite[(7.110)]{jackson}) that for $t \in \R$ one has
\begin{equation}\label{eq:LorentzModelTimeDomain}
\chi(t) = \sqrt{2 \pi} \theta(t) \sum_{\ell=1}^n \omega_{p,\ell}^2\, \frac{\sin(\nu_\ell t)}{\nu_\ell}\, e^{-\gamma_\ell t }\, {I_3}
\quad \text{and} \quad \lambda(t) = \sqrt{2 \pi} \theta(t) \sum_{\ell=1}^n \omega_{p,\ell}^2\, \frac{d}{dt} \left( \frac{\sin(\nu_\ell t)}{\nu_\ell}\, e^{-\gamma_\ell t } \right) {I_3},
\end{equation}
where $\nu_\ell^2 = \omega_{0,\ell}^2-\gamma_\ell^2$ (if $\omega_{0, \ell} > \gamma_\ell$) and  $\theta$ is the Heaviside function, i.e., $\theta(t) = 1$ if $t \ge 0$ and $\theta(t) = 0$ otherwise. Here $\lambda$ is defined in such a way that $\widehat{\lambda}(\omega) = -  i \omega \widehat{\chi}(\omega)$ for $\omega \in \R$.

\medskip 
We study \eqref{eq:EM-G-1} under the form of the initial problem at the time $t = 0 $, assuming that the data are known in the past $t<0$. 
Set 
\begin{equation}
\label{eq:truncatedConvolutionProduct}
(\lambda_{ij} \star  X)(t,\cdot) := \int_{0}^t \lambda (t-\tau,\cdot) X(\tau,\cdot)\,   d \tau, \quad  \mbox{ for } t > 0. 
\end{equation}
For $X = E$ or $H$, under the causality assumption \eqref{eq:causality}-\eqref{eq-causality1}, one has for $t>0$ that
\begin{equation}\label{reformulate}
\begin{aligned}
(\lambda_{ij} *  X)(t,\cdot)&  =  \int_{0}^t \lambda_{ij} (t-\tau,\cdot) X(\tau,\cdot)\,  d \tau + \int_{-\infty}^0 \lambda_{ij} (t-\tau,\cdot) X(\tau,\cdot)\,  d \tau  \nonumber \\
& = (\lambda_{ij} \star  X)(t,\cdot) + \int_{-\infty}^0 \lambda_{ij} (t-\tau,\cdot) X(\tau,\cdot)\,  d \tau.
\end{aligned}
\end{equation}
Hence if  the data are known for the past $t < 0$, then the last term is known at time $t > 0$. 
With the presence of sources,  one can then reformulate system   \eqref{eq:EM-G-1} under the form
\begin{equation}\label{eq:problem0}
\left\{ 
\begin{aligned}
& \eps_{rel}(x) \partial_t E(t, x) + (\lambda_{ee} \star E)(t, x) +  (\lambda_{em} \star H)(t, x) =  \nabla \times H(t, x)+f_e(t,x), \\[6pt]
&  \mu_{rel}(x)  \partial_t H(t, x)+  (\lambda_{me}\star E)(t, x) + (\lambda_{mm}\star H)(t,x)   = - \nabla \times E(t, x)+f_m(t,x),\\[6pt]
& E(0,x) = E_0(x),  \ H(0,x) = H_0(x), 
\end{aligned}\right.
\end{equation}
for $t > 0$ and $x \in \R^3$. Here $E_0$ and $H_0$ are the initial data at time $t =0$, and $f_e, \ f_m$ are given fields that can be considered as ``effective" sources since they also take into account the last terms in \eqref{reformulate}. 

Set
\begin{equation}\label{eq:notations}
u := \begin{bmatrix}
E \\ H
\end{bmatrix}, \quad 
u_0 := \begin{bmatrix}
E_0 \\ H_0
\end{bmatrix}, \quad
f := \begin{bmatrix}
f_e \\ f_m
\end{bmatrix}, \quad
\mathbb A u := \begin{bmatrix}
\nabla \times H \\ - \nabla \times E
\end{bmatrix},
\end{equation}
\begin{equation}\label{eq:notations-1}
\Lambda := \begin{bmatrix}
\lambda_{ee} & \lambda_{em} \\
\lambda_{me} & \lambda_{mm} 
\end{bmatrix} \quad \mbox{ and } \quad 
M := \begin{bmatrix}
\eps_{rel} & 0 \\
0 & \mu_{rel} 
\end{bmatrix}.
\end{equation}
System \eqref{eq:problem0} can then be rewritten in the following compact form: 
\begin{equation}\label{eq:problem}
\left\{ 
\begin{aligned}
& M(x) \partial_t u(t,x) + (\Lambda \star u)(t,x) = \mathbb A u(t,x) + f(t,x), \\
& u(0,x) = u_0(x),  \\
\end{aligned}\right.  \quad  \mbox{ for } t> 0,\  x \in \R^3.
\end{equation}

Define
\begin{equation}
\HH := \LL^3 \times \LL^3 \quad \mbox{ and } \quad \VV := \Hcurl \times \Hcurl, 
\end{equation}
equipped with the standard inner products induced from $\LL^3$ and $\Hcurl$.  One can verify that $\HH$ and $\VV$ are Hilbert spaces. 

We also denote 
\begin{equation}\label{def-M6}
\mbox{$\mathcal M_6( L^\infty(\R^3))$ as the space of $6 \times 6$ real matrices whose entries are $ L^\infty(\R^3)$ functions.}
\end{equation}
In what follows, in the time domain, we only consider {\it real} quantities. 

\medskip
Concerning  the well-posedness of \eqref{eq:problem}, we prove, see \cite[Theorem 3.1]{Ng-Vinoles},

\begin{theorem}\label{thm-WP}
Let $T \in (0,+\infty)$, $u_0  \in \HH$, $f \in {L}^1(0,T;\HH)$,  and $\Lambda \in  L^1\big(0,T;\mathcal M_6( L^\infty(\R^3) \big)$. Assume that \eqref{epsmu}, \eqref{assumption-lambda}, \eqref{eq:causality} and
\eqref{eq:passivity} hold.  There exists a unique weak solution $u \in  L^\infty(0,T;\HH)$  of \eqref{eq:problem} on $(0,T)$. Moreover, the following estimate holds 
\begin{equation}\label{est-WP}
\ll M u(t, \cdot), u(t, \cdot) \rr_\HH  \le  \left(\ll M u_0, u_0 \rr_\HH^{1/2}   + C \int_0^t \|  f(s, \cdot) \|_{\HH} \,   ds\right)^2 \quad \mbox{ in } (0, T), 
\end{equation}
where $C$ is a positive constant depending only on the coercivity of $M$. 
\end{theorem}

The notion of weak solutions for \eqref{eq:problem} is: 

\begin{definition}\label{def-WS} \rm
Let $T \in (0,+\infty)$, $u_0 \in \HH$ and $f \in  L^1(0,T;\HH)$. A function  $u \in  L^\infty(0,T;\HH)$ is called a \emph{weak solution} of \eqref{eq:problem} on $[0,T]$ if 
\begin{equation}\label{eq-WS}
\frac{ d}{ dt} \ll M u(t, \cdot ),v \rr_\HH + \ll (\Lambda \star u)(t, \cdot),v\rr_\HH   =  \ll u(t, \cdot), \mathbb A v \rr_\HH  + \ll f(t, \cdot),v \rr_\HH \mbox{ in } (0, T)  \mbox{ for all $v \in \VV$},
\end{equation}
and \begin{equation}\label{eq-IC}
u(0, \cdot) = u_0. 
\end{equation}
\end{definition}


\begin{remark} \rm One can easily check that if $u$ is a smooth solution and decays sufficiently at infinity, then $u$ is a weak solution by integration by parts, and that if $u$ is a weak solution and smooth, then $u$ is a classical solution. 
\end{remark} 

%

We next discuss the finite speed propagation for \eqref{eq:problem}. 
In what follows, $B(a,R)$ stands for the ball in $\R^3$ of radius $R>0$ and centered at $a \in \R^3$. In the case $a = 0$ -- the origin -- we simply denote $B(0, R)$ by $B_R$. 
Set  
\begin{equation}
c(x) := \gamma_e (x) \gamma_{m} (x),   \quad  \mbox{ for } x \in \R^3, 
\end{equation}
where $\gamma_e(x)$ and $\gamma_m(x)$ are the largest eigenvalues of $\eps_{rel}(x)^{-1/2}$ and $\mu_{rel}(x)^{-1/2}$, respectively. According to assumption \eqref{epsmu}, $c(x)$ is bounded above and below  by a positive constant. For $a \in \R^3$ and $R> 0$, we denote  
\begin{equation}\label{eq:defSpeed}
c_{a, R}:= \mathop{\mbox{ess sup}}_{x \in B(a, R)} c(x).   
\end{equation}

The following result  is on the  finite speed propagation of \eqref{eq:problem}, see \cite[Theorem 3.2]{Ng-Vinoles}:

\begin{theorem}\label{thm-FS} Let $R>0$, $a \in \R^3$, and $u_0  \in \HH$. For $T > R/c_{a,R}$, let $f \in  L^1(0, T;\HH)$ and $\Lambda \in  L^1(0, T;\mathcal M_6\big( L^\infty(\R^3)\big)$. Assume that \eqref{epsmu}, \eqref{assumption-lambda}, \eqref{eq:causality} and
\eqref{eq:passivity} hold, 
\begin{equation}
{supp}\, u_0 \cap B(a,R)  = \emptyset,
\end{equation}
and
\begin{equation}
{supp}\, f(t, \cdot) \cap B(a,R-c_{a, R} t) = \emptyset, \quad \mbox{ for almost every } t \in (0, R/ c_{a, R}). 
\end{equation}
Let $u \in  L^\infty(0, T; \HH)$ be  the unique weak solution of \eqref{eq:problem} on $(0, T)$. Then 
\begin{equation}
{supp}\, u(t, \cdot) \cap  B(a,R-c_{a, R} t) = \emptyset, \quad \mbox{ for almost every } t \in (0, R/ c_{a, R}). 
\end{equation}
\end{theorem}

We briefly mention here the ideas of the proofs of Theorems~\ref{thm-WP} and \ref{thm-FS}. The construct of a solution in Theorem~\ref{thm-WP} is based on the Galerkin method. One of the key observations is the following inequality 
\begin{equation}\label{key-Galerkin}
 \int_0^t  \ll (\Lambda \star v)(s, \cdot),  v(s, \cdot) \rr_\HH  d s \ge 0, \quad \mbox{ for } v \in L^\infty(0, T; \HH), \ t \in (0,T). 
\end{equation}
Similar observations in the acoustic setting were used in different contexts, see, e.g., \cite{MinhLinh, NguyenVogelius}. 
The inequality \eqref{key-Galerkin} plays an important role in deriving the following estimate for an approximate  solution $u_n$ after multiplying the equation of $u_n$ by $u_n$ and integrating by parts, which gives
\begin{equation}\label{eq:tmp3}
\ll M u_n(t, \cdot), u_n(t, \cdot) \rr_\HH 
\le \ll M u_{n}(t=0, \cdot), u_{n}(t=0, \cdot) \rr_\HH  + 2 \int_0^t \|  f(s, \cdot) \|_{\HH} \| u_{n}(s, \cdot) \|_{\HH}\,   ds.
\end{equation}
By Gronwall's lemma, this in turn implies the desired estimate for a solution $u$ obtained via the standard compactness argument, see,  e.g.,  \cite{Evans}. . The uniqueness of $u$ is quite standard as in the standard wave equations after noting \eqref{key-Galerkin}.  The proof of Theorem~\ref{thm-FS} is standard via \eqref{key-Galerkin} if one knows that  the solution $u$ is regular. To overcome the lack of the regularity of $u$, we consider the function 
$$
U(t, x): = \int_0^t u(s, x) \, ds, \mbox{ for } t \in [0, T), x \in \R^3
$$
and show that 
\begin{equation}
{supp}\, U(t, \cdot) \cap  B(a,R-c_{a, R} t) = \emptyset, \quad \mbox{ for almost every } t \in (0, R/ c_{a, R}). 
\end{equation}
This yields the desired conclusion. As far as we know, the proof of finite speed propagation for energy solutions is not presented in standard references on partial differential equations.

\section{Other topics and future directions} \label{sect-perspectives}

Some interesting aspects of NIMs are not discussed in this survey, such as the stability of NIMs and cloaking a source via anomalous localized resonance, because we have nothing new to add to these topics.  The stability of NIMs in the frequency domain for acoustic waves was investigated by Costabel and Stephan in 1985 \cite{CostabelErnst} using the integral method. Later, this problem was studied by the integral method and the pseudo-differential operators theory \cite{Ola} and by the $T$-coercivity approach (see \cite{AnneSophieChesnelCiarlet1,AnneSophie-Ciarlet01} and references therein). In these works, the well-posedness was established in the Fredholm sense in $H^1$, meaning  that the compactness holds; the existence and the uniqueness are not discussed. Recently,  \cite{Ng-WP} we introduced a new approach to study the stability aspect of NIMs. More precisely, we investigated the well-posedness  of the Helmholtz equations involving sign changing coefficients. Our approach involved   the study of Cauchy problems, which are derived by  reflections in the spirit of the proofs presented in Sections~\ref{sect-SCM}, \ref{sect-CCM}, and \ref{sect-CALR-object} using the change of variables formula in Lemma~\ref{lem-TO}. We then proposed various methods to study these Cauchy problems. One method was  via the prominent work of Agmon, Douglis, and Nirenberg \cite{ADNII} (via Fourier analysis or fundamental solutions) and others were based on variational methods/ multiplier techniques.  In consequence, we can unify and extend largely  known works.  In particular, we proved that (see \cite[Corollary 1]{Ng-WP}) the well-posedness holds if, under some smoothness assumptions, 
$$
A_+ > A_- \mbox{ on } \Gamma \mbox{ or } A_+ < A_- \mbox{ on } \Gamma, 
$$
for all  connected component $\Gamma$ of the sign changing coefficient interface, $A_+$ is the restriction of $A$ in the region $A>0$, and $A_-$ is the restriction of $-A$ in the region $A < 0$. We also showed that the complementary property of media is almost necessary for the occurrence of resonance (see \cite[Proposition 2]{Ng-WP}). A numerical algorithm in the spirit of  this approach was also studied in \cite{Abdulle16}. 

 The second aspect we do not discuss in this survey is cloaking {\it a source} via anomalous localized resonance. This cloaking technique is relative due to the fact that the power, which is roughly speaking the standard energy of the fields in the region of NIMs multiplied by  the loss,  must be normalized for the cloaking purpose. This phenomenon  was observed by Milton and Nicorovici in \cite{MiltonNicorovici} (see also \cite{NicoroviciMcPhedranMiltonPodolskiy1}) for a symmetrical radial structures in a two dimensional quasistatic regime and was considered  in a general setting, the setting of  doubly complementary media in \cite{Ng-CALR, Ng-CALR-frequency} for the acoustic regime (see also \cite{BouchitteSchweizer10, AmmariCiraoloKangLeeMilton2, KohnLu} for related results in some specific  settings).  It has been shown \cite{Ng-CALR, Ng-CALR-frequency} that $i)$ cloaking a source via anomalous localized resonance appears if and only if the power blows up; $ii)$ the power blows up if the source is located  ``near" the  plasmonic layer made of NIMs; $iii)$ the power remains bounded if the source is far away from the plasmonic layer.  It is worth noting that there is no connection between the blow up of the power and the localized resonance in general  \cite{MinhLoc1}.  Finally, we want to mention that the design of metamaterials poses new and  interesting problems that are being extensively investigated  in the litterature, see \cite{CL, BouchitteFelbacq,GuyenneauZolla,KohnShipman} and the refences therein. 

An interesting direction concerning NIMs, or more generally metamaterials is to study these metamaterials in the time domain.  For example, it would be interesting to understand conditions under which the energy of solutions of the Maxwell equations considered in Section~\ref{sect-Maxwell} decay in any bounded domain; this is known for  (standard) positive index media. Another interesting question would be to investigate the limiting amplitude principle, which concerns the behavior of the fields in the time domain generated by a harmonic forcing term for large time. In some particular settings, the limiting amplitude principle was already considered in \cite{GralakTip, Joly16}, but the question for a general setting remains open.

\bigskip
\noindent{\bf Acknowledgement:}  This paper is an extended  version  of the lecture given by  the author at 
VIASM annual meeting in 2017 at Vietnam Institute for Advanced Study in Mathematics. The author warmly thanks  the institute for the hospitality. 


\end{document}